\documentclass[10pt,reqno]{amsart}
\usepackage{amsthm,amsfonts,amsmath,latexsym,amssymb,bbold,dsfont}
\usepackage{amssymb,latexsym,amsfonts,amsmath,amsthm,mathabx}
\usepackage[pdfencoding=auto]{hyperref}
\usepackage{verbatim}
\usepackage[integrals]{wasysym}

\usepackage[dvips]{graphics}
\usepackage{graphicx,color}
\usepackage{epsfig} 

\usepackage[rightcaption]{sidecap}

\usepackage{mathtools}
\usepackage{mfirstuc}
\usepackage[normalem]{ulem}
\usepackage{enumerate} 
\usepackage{cite}
\def\a{\alpha}
\def\b{\beta}
\def\g{\gamma}
\def\d{\delta}
\def\e{{\rm e}}

\def\D{\Delta}
\def\G{\Gamma}
\def\O{\Omega}
\def\L{\Lambda}
\def\p{\partial}

\def\C{\mathbb C}

\def\Z{\mathbb Z}
\def\R{{\mathbb R}}

\def\N{\mathbb N}
\def\tr{{\rm tr}\,}

\def \Lp {\mathsf L}  

\usepackage{xparse,aliascnt,hyperref,bookmark}

\NewDocumentCommand{\xnewtheorem}{m o m}
 {%
  \IfNoValueTF{#2}
   {\newtheorem{#1}{#3}}
   {%
    \newaliascnt{#1}{#2}%
    \newtheorem{#1}[#1]{#3}%
    \aliascntresetthe{#1}%
    \expandafter\newcommand\csname #1autorefname\endcsname{\makefirstuc  {\lowercase {#3}}}%
   }%
 }

\hypersetup{
  colorlinks=true,
   linkcolor=blue,
   filecolor=magenta,      
   urlcolor=cyan,
}

\theoremstyle{plain}
\newtheorem{thmc}{ERROR}[section]
\iftrue 
\xnewtheorem{thm}[thmc]{THEOREM}
\xnewtheorem{lemma}[thmc]{LEMMA}
\xnewtheorem{cl}[thmc]{COROLLARY}
\xnewtheorem{prop}[thmc]{PROPOSITION}
\xnewtheorem{assumption}[thmc]{ASSUMPTION}
\xnewtheorem{conj}[thmc]{CONJECTURE}
\xnewtheorem{fact}[thmc]{FACT}
\theoremstyle{definition}
\xnewtheorem{defin}[thmc]{DEFINITION}

\theoremstyle{remark}
\xnewtheorem{remark}[thmc]{REMARK}
\xnewtheorem{remarks}[thmc]{REMARKS}

\else

\theoremstyle{plain}
\xnewtheorem{thm}[thmc]{Theorem}
\xnewtheorem{lemma}[thmc]{Lemma}
\xnewtheorem{cl}[thmc]{Corollary}
\xnewtheorem{prop}[thmc]{Proposition}
\xnewtheorem{assumption}[thmc]{Assumption}
\xnewtheorem{conj}[thmc]{Conjecture}
\xnewtheorem{fact}[thmc]{Fact}

\theoremstyle{definition}
\xnewtheorem{defin}[thmc]{Definition}

\theoremstyle{remark}
\newtheorem{remark}[thmc]{Remark}
\newtheorem{remarks}[thmc]{Remarks}
\fi

\newcommand{\be}{\begin{equation}}
\newcommand{\ee}{\end{equation}}
\newcommand{\bea}{\begin{eqnarray}}
\newcommand{\eea}{\end{eqnarray}}
\newcommand{\beax}{\begin{eqnarray*}}
\newcommand{\eeax}{\end{eqnarray*}}

\newcommand{\mfr}[2]{{\textstyle\frac{#1}{#2}}}


\numberwithin{equation}{section}

\begin{document}

\title[Entanglement entropy of free fermions in a magnetic field]{Entanglement entropy of ground states of the three-dimensional ideal Fermi gas in a magnetic field}

\date{September 20, 2022}

\author[P.~Pfeiffer and W.~Spitzer]{Paul Pfeiffer and Wolfgang Spitzer}
\address{Fakult\"at f\"ur Mathematik und Informatik, 
FernUniversit\"at in Hagen, Universit\"atsstra\ss e 1, 
58097 Hagen, Germany}
\email{paul.pfeiffer@fernuni-hagen.de}
\email{wolfgang.spitzer@fernuni-hagen.de}

\begin{abstract} We study the asymptotic growth of the entanglement entropy of ground states of non-interacting (spinless) fermions in $\R^3$ subject to a non-zero, constant magnetic field perpendicular to a plane. As for the case with no magnetic field we find, to leading order $L^2\ln(L)$, a logarithmically enhanced area law of this entropy for a bounded, piecewise Lipschitz region $L\L\subset \R^3$ as the scaling parameter $L$ tends to infinity. This is in contrast to the two-dimensional case since particles can now move freely in the direction of the magnetic field, which causes the extra $\ln(L)$ factor. The explicit expression for the coefficient of the leading order contains a surface integral similar to the Widom formula in the non-magnetic case. It differs however in the sense that the dependence on the boundary is not solely on its area but on the ``area perpendicular to the direction of the magnetic field". On the way we prove an improved two-term asymptotic expansion (up to an error term of order one) of certain traces of one-dimensional Wiener--Hopf operators with a discontinuous symbol. This is of independent interest and leads to an improved error term of the order $L^2$ of the relevant trace for piecewise $\mathsf{C}^{1,\alpha}$ smooth surfaces $\p\L$. 
\end{abstract}

\maketitle

\tableofcontents

\section{Introduction}

In recent years, entanglement entropy (EE) has become an important and intensively studied quantity of states of many-particle quantum systems. For an introduction to this topic we refer to \cite{AFV,ECP,HHHH}. In this paper, we study the EE of ground states of the ideal Fermi gas in a magnetic field in three-dimensional Euclidean space, $\R^3$. The two-dimensional Fermi gas in a constant magnetic field was recently analyzed in \cite{CE} and \cite{LSSmagn}, starting from the earlier work in \cite{RS1}. Here, a strict area-law holds, while for the free Fermi gas in any dimension $d$ a logarithmically enhanced area-law is valid, see \cite{GK,LSS1}. Stability of these area-laws has been proved in \cite{MS2019,MS2020} for $d\ge2$ and in \cite{P2021} in the sense that adding a ``small" electric or magnetic potential to the Hamiltonian does not change the leading asymptotics of the entropy. The one-dimensional case seems to be still open (for the $\a$-R\'enyi entropy with $\a\le1$).

There is an extensive literature on EE by now with many fascinating connections and implications to related fields. Here we only mention and refer to a small fraction of mathematical results. In \cite{PS}, an enhanced area-law was proved for the one-dimensional free Fermi gas in a periodic potential; the higher dimensional case remains an open problem. By the work in \cite{PasSlav,EPS,MPS} we understand EE in Anderson-type models on the lattice. An extension to the EE of positive temperature equilibrium states (of the ideal Fermi gas) was presented in \cite{LSS-temp1,LSS-temp2,Sob2016}. Finally, we mention results on the XY and XXZ quantum spin chain~\cite{AFS,JinKor,BW,EKS,FS,FO}.

By a (strict) area-law for a ground state of the infinitely extended Fermi gas, say in $\R^d$ with spatial dimension $d\in\N$,  we mean that the entanglement (or local) entropy of this state reduced to the scaled (bounded) region $L\L$ grows to leading order like $L^{d-1}\mathcal H^2(\p\L)$ as the dimensionless real parameter $L$ tends to infinity. Here, $\mathcal H^2(\p\L)$ is the (Hausdorff) surface area of the boundary $\p\L$. If there is an extra $\ln(L)$ factor in this leading asymptotics, then we call it a logarithmically enhanced area-law.  

Whether one should expect a strict area-law or an enhanced area-law is related to the spectral properties of the one-particle Hamiltonian of the non-interacting many-particle Fermi gas. If the off-diagonal part of the integral kernel of the corresponding spectral (Fermi) projection has a fast decay (e.g., exponential), then we expect a strict area-law to hold. It is not difficult to argue for that (see \cite{PasSlav}) but to compute and finally prove the precise leading coefficient has only been accomplished in special cases. On the other hand, if the decay of the off-diagonal part of the integral kernel is weak (e.g., inverse linear), then we can expect an enhanced area-law. In the present model we have a mixture. Namely we have an exponential decay in the planar coordinate (orthogonal to the magnetic field) and a $1/|\cdot|$ decay in the longitudinal coordinate along the magnetic field. The latter prevails and leads to a logarithmically enhanced area-law. Our main result is formulated in \autoref{entropy of ground state}. 

As in previous proofs there is two parts to proving such a result. Firstly, we prove a two-term asymptotic expansion for polynomials (see \autoref{thm:polynomial}). Due to the product structure of the ground state, see \eqref{Dmu exp}, we can dimensionally reduce the asymptotics of a three-dimensional problem to an asymptotic expansion of a one-dimensional problem with localizing sets $L\L_{x^\perp}\subset \R$ and with the spectral projection of the one-dimensional Laplacian, see \autoref{lem:reduction}. The corresponding asymptotic expansion was already proved by Landau and Widom~\cite{LW} and then improved by Widom~\cite{Widom1982}. But here we need to take care of the error term which depends on the planar coordinate $x^\perp\in\R^2$ and integrate over $x^\perp$. To this end, we show that the error term is of order one and is integrable as a function of $x^\perp$ under some assumptions on $\L$.  We believe that the precise description of the error term for the one dimensional free case in terms of the finite collection of intervals $\L_{x^\perp}$ is of independent interest and we provide a proof in Appendix \ref{Appendix C}. This dimensional reduction is also the strategy of Widom in \cite{Widom1990} and of Sobolev in the proof of the Widom conjecture in \cite{Sob:AMS}. In fact, due to the fast (exponential) decay in the planar direction error estimates are simpler to obtain than in the case with no magnetic field. This and the improved Landau--Widom (or Widom) asymptotics allows us to prove for $\mathsf{C}^{1,\a}$ (smooth) regions $\L$ an error term (for polynomials as in \autoref{thm:polynomial}) of the order $L^2$ rather than merely of lower order than $L^2\ln(L)$ in \cite[Theorem 2.9]{Sob:AMS}. 

Secondly, in Section \ref{Section:Entanglement} we make the transition in the asymptotic expansion from polynomials to the entropy function. This requires certain Schatten--von Neumann quasi-norm bounds presented in Section \ref{Section: Schatten}, which in turn are based on bounds obtained in previous papers~\cite{LSSmagn,LSS1} and notably by Sobolev~\cite{Sob:Schatten}.

The smoothness conditions on the region $\L$ to prove our two-term asymptotic result with error term $o(L^2\ln(L))$ are rather weak, namely we require $\L$ to be only piecewise Lipschitz smooth. For a smooth region $\L$, one would expect the next lower order term to be of the order $L^2$. This is indeed true if the boundary $\p\L$ is piecewise $\mathsf C^{1,\alpha}$ smooth. We also present regions with weaker regularity on the boundary for which the error term (for a quadratic polynomial) can be arbitrarily close to the leading $L^2\ln(L)$-term. This may also be of independent interest and is the content of Section \ref{Section: lower order}.

A note on our notation: As $L$, $L\ge1$, is our scaling parameter that tends to infinity, we use the ``big-O" and ``small-o" notation in the sense that for two functions $f$ and $g$ on $\R^+$, $f = O(g)$ if $\limsup_L f(L)/g(L) <\infty$ and $f = o(g)$ if $\limsup_L f/g(L) =0$. By $C$ with or without indices, we denote various positive, finite constants, whose precise values is of no importance, and may even change from line to line.

\section{Setup}

We consider a non-zero, constant magnetic field in $\R^3$ of strength $B$ which is perpendicular to a plane. We assume without loss of generality that this constant magnetic field points in the positive $z$-direction with $B>0$. 

We denote the Euclidean norm in $\R^d$, $d\in\N$, or the norm in the Hilbert-space $\text L^2(\R^d)$ of complex-valued, square-integrable functions on $\R^d$ by the same symbol $\|\cdot\|$. For $x \in \R$, let $\langle x \rangle\coloneqq \sqrt{ 1+ x^2}$ denote the Japanese bracket. For a Borel set $\Omega \subset \R^d$ and $k <d$, let $\mathcal H^k(\Omega)$ be be the $k$-dimensional Hausdorff measure of $\Omega$, $\#\Omega=\mathcal H^0(\Omega)$ its counting measure, and let $\lvert \Omega \rvert$ be its $d$-dimensional Lebesgue measure/volume. By $\mathds{1}_{\Omega}$ we denote the multiplication operator on $\Lp^2(\R^d)$ by the indicator function $1_\Omega$ of the set $\Omega$. As usual, we write for the complement $\Omega^\complement \coloneqq \R^d \setminus \Omega$.  

For $r>0$, $x\in\R^d$, and a set $X\subset \R^d$ we denote by 
\be \label{def: ball} B_r(x) \coloneqq \big\{y\in\R^d : \|y-x\|<r\big\}\,,\quad B_r(X) \coloneqq X+B_r(0)\coloneqq \big\{x+y: x\in X, \|y\| <r\big\}
\ee
the open ball of radius $r$ with center $x$ and the (open) \emph{$r$-neighborhood} of the set $X\subset \R^{d}$ of width $r$, respectively. In most cases the dimension, $d$, is clear from the context and we omit it in the definition; if not, we write $B^{(d)}_r(x)$. We denote the closed ball of radius $r$ with center $x$ by $\overline{B}_r^{(d)}(x)$.

For a point $x\in\R^3$ we write $x = (x^\perp,x^\parallel)$ with (planar coordinate) $x^\perp\in\R^2$ and (longitudinal coordinate) $x^\parallel\in\R$, and $\nabla = (\nabla^\perp,\nabla^\parallel)$, where $\nabla^\perp$ and $\nabla^\parallel$ are the gradients in the respective Cartesian coordinates. 

By our assumption, the magnetic field is equal to $B\cdot e_3$ with $e_3 \coloneqq (0,0,1)$. As in \cite{LSSmagn}, we use the symmetric gauge $a:\R^2\to\R^2$ defined as $a(x^\perp) \coloneqq B/2\,(x^\perp_2,-x^\perp_1)$ so that the rotation
\be \nabla\times (a,0) = B\cdot e_3\,.
\ee
The one-particle Hamiltonian of the ideal Fermi gas in three-dimensional Euclidean space $\R^3$ subject to the magnetic field $B\cdot e_3$ is informally given by
\be \mathrm{H}_B \coloneqq (-\mathrm{i}\nabla^\perp - a)^2 + (-\mathrm{i}\nabla^\parallel)^2\,.
\ee
It is well-defined as a self-adjoint operator on a suitable domain in the one-particle Hilbert space $\mathsf{L}^2(\R^3)$. 

The ground state of free fermions with one-particle Hamiltonian $\mathrm{H}_B$ is described by the spectral projection (or Fermi projection) $\mathrm{D}_\mu \coloneqq \mathds{1}(\mathrm{H}_B\le\mu) \coloneqq 1_{(-\infty,\mu]}(\mathrm{H}_B)$ of $\mathrm{H}_B$ below some so-called Fermi energy (or chemical potential) $\mu\in\R$. As is well-known, we have \cite{Fock,Landau} 
\be (-\mathrm{i}\nabla^\perp - a)^2 = B\sum_{\ell=0}^\infty (2\ell +1) \mathrm{P}_\ell
\ee 
with explicitly known (infinite-dimensional) eigenprojections $\mathrm{P}_\ell$ on $\mathsf{L}^2(\R^2)$. In order to write down these projections, let us introduce the Laguerre polynomials, $\mathcal{L}_\ell(t) \coloneqq \sum_{j=0}^\ell \frac{(-1)^j}{j!}\, \binom{\ell}{\ell - j}\, t^j$, $t\ge0$, of degree $\ell\in\N_0$. Then the integral kernel of $\mathrm{P}_\ell$ is given by the function
\be \label{Landau_kernel:eq}
p_\ell(x^\perp,y^\perp) \coloneqq \frac{B}{2\pi} \,\mathcal{L}_\ell\big(B\|x^\perp-y^\perp\|^2/2\big)\,\exp\big(-B\|x^\perp-y^\perp\|^2/4+ \mathrm{i}\mfr{B}2 x^\perp \wedge y^\perp\big)\,,\quad x^\perp,y^\perp\in\mathbb R^2\,.
\ee 
Here, $\wedge$ refers to the exterior or wedge product on $\R^2$. The explicit description of this kernel is not relevant for this paper. We only use the exponential decay in $\lVert x^\perp-y^\perp \rVert^2$ and $p_\ell(x^\perp,x^\perp)=B/(2 \pi)$.
In the $z$-direction, we meet the spectral projection $\mathds{1}((-\nabla^\parallel)^2\le\mu)$ with (sine) integral kernel, $\mathds{1}((-\nabla^\parallel)^2\le\mu)(z,z')=k_\mu(z-z')$, 
\begin{align} \label{k mu}
k_\mu(z)\coloneqq \begin{cases}
 \frac{\sin(\sqrt{\mu}z)}{\pi z}\,, &\mbox{ for } z\in\R \setminus \{0\} \\
 \lim_{z \to 0} k_{\mu}(z) = \frac{\sqrt \mu} \pi &\mbox{ for } z=0 \end{cases}
 \quad \mu>0\,.
\end{align}

The following factorization of spectral projections is crucial, which stems from the fact that the magnetic field is pointing in the $z$-direction. We work with the identification $\Lp^2(\R^2) \otimes \Lp^2(\R)=\Lp^2(\R^3)$. Since the spectrum of $\mathrm{H}_{B}$ is the set $[B,\infty)$, we may always consider $\mu> B$ since for smaller values of $\mu$ the ground state is zero. If $B<\mu\le 3B$ then $\mathrm{D}_\mu = \mathrm{P}_0 \otimes \mathds{1}[(-\mathrm{i}\nabla^\parallel)^2 \le \mu - B]$. For higher values of $\mu$, let $\nu\coloneqq\lceil\frac12(\mu/B -1)\rceil\in\N$ be the smallest integer larger or equal to $\frac12(\mu/B -1)$, and let us set $\mu(\ell) \coloneqq\mu- B(2\ell+1)$. Then
\be 
\mathrm{D}_\mu = \mathds{1}(\mathrm{H}_{B}\le\mu) = \sum_{\ell=0}^{\nu-1} \mathrm{P}_\ell \otimes \mathds{1}[(-\mathrm{i}\nabla^\parallel)^2 \le \mu(\ell)] \label{Dmu exp}
\ee
with integral kernel ($x=(x^\perp,x^\parallel), y=(y^\perp,y^\parallel)$) 
\be \mathrm{D}_\mu(x,y) = \sum_{\ell=0}^{\nu-1} p_\ell(x^\perp,y^\perp) k_{\mu(\ell)}(x^\parallel-y^\parallel)\,. \label{int kernel Dmu}
\ee
For any Borel subset $\L\subset\R^3$ we define the spatial reduction (or truncation) of $\mathrm{D}_\mu$ to $\L$ by
\be \mathrm{D}_\mu(\L) \coloneqq \mathds{1}_\L \mathrm{D}_\mu \mathds{1}_\L\,. \label{def Dmu L}
\ee 

Before we define the main object in this paper, we introduce for any $\g>0$ the $\g$-R\'enyi entropy function, $h_\g:[0,1]\to[0,\ln(2)]$, 
\begin{align} \label{def:entropy}
h_\g(t) \coloneqq &\ \frac{1}{1-\g}\ln\big(t^\g + (1-t)^\g\big)\,,\ \g\not = 1\,,\\[0.2cm]
h_1(t) \coloneqq &\ -t\ln(t) - (1-t)\ln(1-t)\mbox{ if } t\not\in\{0,1\} \mbox{ and } h_1(0)\coloneqq h_1(1)\coloneqq 0\,. 
\end{align}
Now, for a ground state described by the projection $\mathrm{D}_\mu=\mathds{1}(\mathrm{H}_B\le\mu)$ as above, a Borel subset $\L\subset\R^3$, and localized ground-state projection, $\mathrm{D}_\mu(\L) $, we define the $\g$-R\'enyi entanglement entropy of the ground state at Fermi energy $\mu$ localized (in space) to $\L$ by
\be 
\mathrm S_\g(\L) \coloneqq \tr h_\g(\mathrm{D}_\mu(\L))\,.
\ee
Here, $\tr$ refers to the (usual Hilbert space) trace on $\mathsf{L}^2(\R^d)$. For bounded $\L$, $h_\g(\mathrm{D}_\mu(\L))$ is trace-class by the same arguments as in the proof of Lemma 7 in \cite{LSSmagn}; thus,  the entanglement entropy $\mathrm S_\g(\L)$ is trivially a positive number. This entropy is a rather complicated function of $\L$, but there is a chance to describe it for large regions. To this end, we scale a fixed set $\L$ by $L, L \ge 1$, and we determine the leading growth (scaling) of the entropy $\mathrm S_\g(L\L)$ as $L\to\infty$.

As there does not seem to be a common definition for regions with piecewise differentiable boundary, we will now provide the one used in this paper. 

\begin{defin} \label{def ps} Let $0 < \alpha <1, d\in\N$. A \emph{region} $\Lambda\subset \R^{d+1}$ is a finite union of bounded, open, connected sets in $\R^{d+1}$ such that their closures (denoted by $\bar{\cdot}$) are disjoint. The boundary $\p\L$ is the set $\bar{\L}\setminus\L$. We assume that the closures $\overline{\L}$ and $\L^\complement$ are topological manifolds with boundary $\p \L$.

We call a bi-Lipschitz\footnote{\label{biLip def} A function $f$ is bi-Lipschitz if there is a constant $C_{\text{lip}} \in \R^+$ such that $C_{\text{lip}}^{-1}\lVert x-y \rVert \le \lVert f(x)-f(y) \rVert \le C_{\text{lip}} \lVert x-y \rVert$. Such a function $f$ is (obviously) invertible on its image and satisfies $C_{\text{lip}}^{-1}\lVert x-y \rVert \le \lVert f^{-1}(x)-f^{-1}(y) \rVert \le C_{\text{lip}} \lVert x-y \rVert$. } map $\Psi \colon [0,1]^d \to \partial \L$ a Lipschitz chart of $\p \L$ if $\Psi((0,1)^d) \subset \p\L$ is relatively open. If in addition $\Psi \in\mathsf{C}^1((0,1)^d)$ and its differential $D\Psi$ satisfies the H\"older condition
\be \lVert D\Psi(x) - D\Psi(y) \rVert \le C \lVert x-y \rVert^\alpha \,, \quad x,y \in (0,1)^d \,, \label{C1a cond} \ee
for some constant $C$,  we say that $\Psi$ is a $\mathsf{C}^{1,\alpha}$ chart. A finite set of charts $(\Psi_i)_{i \in I}$ is called a piecewise atlas  of $\p \L$ if $ \p \L=\bigcup_{i\in I} \Psi_i([0,1]^d)$, and a global atlas  of $\p \L$ if $ \p \L=\bigcup_{i\in I} \Psi_i((0,1)^d)$.  We say an atlas is  a Lipschitz  atlas (resp.~$\mathsf{C}^{1,\a}$) if it consists of Lipschitz (resp.~$\mathsf{C}^{1,\a}$) charts.

We say that $\L$ is a \textit{piecewise Lipschitz region} (resp. \textit{global Lipschitz region}) if $\p \L$ admits  a piecewise Lipschitz atlas $(\Psi_{\text{pL},i})_{i \in I}$  (resp.~global Lipschitz  atlas $(\Psi_{\text{gL},i})_{i \in I}$). We call $\L$ a \textit{piecewise $\mathsf{C}^{1,\a}$ region} if it admits both a global Lipschitz atlas $(\Psi_{\text{gL},j})_{j \in J}$ and a piecewise $\mathsf{C}^{1,\a}$ atlas $(\Psi_{\text{pC},i})_{i \in I}$.

For a \textit{piecewise $\mathsf{C}^{1,\a}$ region} $\L$,  we fix a piecewise  $\mathsf{C}^{1,\alpha}$ atlas $(\Psi_{\text{pC},i})_{i \in I}$ and define the set of all edges, $\Gamma$ by
\be \Gamma  \coloneqq \bigcup_{i \in I}\Psi_{\text{pC},i}(\partial ( [0,1]^d)) \,. \ee
\end{defin}
\begin{remarks}
\begin{enumerate}
\item[(i)] Any \textit{global} Lipschitz region is obviously a \textit{piecewise} Lipschitz region.

\item[(ii)] Our definition of a \textit{global Lipschitz region} is a bit more general than the usual notion of a \emph{strong Lipschitz region} (see \cite[Pages 66--67]{Adams}), where every $v \in \p\L$ has a neighborhood $U_v \subset \R^{d+1}$ such that, after an affine-linear transformation, the set  $\L \cap U_v$ looks like the graph below a Lipschitz function $\Psi_v \colon (0,1)^d \to \R$. To get to our definition from this, one can choose the graph function $x \mapsto (x, \Psi_v(x))$ on $(0,1)^d$ as the bi-Lipschitz function needed in our definition. (As a Lipschitz function, it naturally extends to all of $[0,1]^d$.) 

\item[(iii)]
For a piecewise Lipschitz region $\L\subset\R^{d+1}$ and for $v \in \p \L$, let $n(v)$ be the unit outward normal vector at $v$. This is only well defined up to null sets with respect to the $d$-dimensional Hausdorff (surface) measure $\mathcal H^{d}$ on $\p\L$, see \autoref{null set1}.

\item[(iv)] As the set of edges, $\Gamma$, depends on the piecewise $\mathsf{C}^{1,\alpha}$ atlas $\Psi_{\mathrm{pC},i}$ it may be a different set depending on the atlas. 
\end{enumerate}
\end{remarks}

For a continuous function $f:[0,1]\to\C$ with $f(0)=0$ and being H\"older continuous at the two endpoints $0$ and $1$, we introduce the linear functional 
\be \label{def: I}
f\mapsto \mathsf I(f) \coloneqq \frac{1}{4\pi^2}\int_0^1\mathrm dt\, \frac{f(t)-tf(1)}{t(1-t)}\,.
\ee
By our assumption, $|\mathsf I(f)|<\infty$. We note for later use two special cases. Namely, $\mathsf{I}(m)\coloneqq \mathsf I((\cdot)^m) = -1/(4\pi^2)\,\sum_{r=1}^{m-1}r^{-1}$; as usual we interpret the sum on the right-hand side as zero if $m=1$, which coincides with the vanishing of $\mathsf{I}$ on affine linear functions. The second example concerns the $\g$-R\'enyi entropy function $h_\g$ defined in \eqref{def:entropy}. Here, $\mathsf I(h_\g) = (1+\g)/(24\g)$, see~\cite{LSS1}. 

Our first main result is the following theorem, which we prove in the next section.

\begin{thm} \label{thm:polynomial} Let $f:[0,1]\to\C$ be a polynomial with $f(0)=0$, let $\L \subset \R^3 , \mu>B>0, \nu \coloneqq\lceil\frac12(\mu/B -1)\rceil\in\N$, the smallest integer larger or equal to $\frac12(\mu/B -1)$, and $\mu(\ell) \coloneqq\mu- (2\ell+1)B$. Let $ \mathrm{D}_\mu (L\L) $ be the operator defined in \eqref{def Dmu L}. 
\begin{enumerate}[(i)]
\item{ 
If $\L$ is a piecewise Lipschitz region (see \autoref{def ps}), then we have the asymptotic expansion of the trace on $\mathsf{L}^2(\R^3)$,
\begin{align} \label{main formula}
\tr f(\mathrm{D}_\mu({L\L})) &= L^3 \frac{B}{2\pi^2} \,\sum_{\ell=0}^{\nu-1} \sqrt{\mu(\ell)} f(1) |\L| \notag
\\
& +  L^2\ln(L) \nu B \,\mathsf{I}(f) \,\frac{1}{\pi} \int_{\p\L} \mathrm{d} \mathcal H^2(v) \,\lvert n(v) \cdot e_3 \rvert + o(L^2\ln(L))\,,
\end{align}
as $L\to\infty$. Here, $n(v)$ is the unit normal outward vector at $v\in\p\L$, which is well-defined for almost every $v \in \p \L$, and $\mathcal H^2$ is the two-dimensional (surface) Hausdorff measure on $\p\L$.}
\item{If $\L$ is a piecewise $\mathsf{C}^{1,\alpha}$ region (see \autoref{def ps}), then the error term is $O(L^2)$ instead of $o(L^2\ln(L))$.} 
\end{enumerate}
\end{thm}

\begin{remarks} \begin{enumerate}
\item[(i)] The condition $f(0)=0$ is no restriction in the sense that in general the operator on the left-hand side has to be replaced by $f(\mathrm{D}_\mu({L\L})) - f(0)\mathrm{D}_\mu({L\L})$ and $\mathsf{I}(f)$ on the right-hand side by $\mathsf{I}(\tilde{f})$ with $\tilde{f}(t):= f(t) - (1-t)f(0)$.  

\item[(ii)] For the ideal Fermi gas with one-particle Hamiltonian $\mathrm{H}_0 = -\Delta$ on $\mathsf{L}^2(\R^3)$, Fermi energy $\mu>0$, ground state Fermi projection $\mathrm{D}_\mu = \mathds{1}(-\Delta\le\mu)$ and Fermi sea $\G \coloneqq \{p\in\R^3 : p^2\le \mu\}$ it was proved in \cite{LSS1} that
\be \label{LSSformula}\tr f(\mathrm{D}_\mu(L\L)) = L^3 f(1) \lvert\G/(2\pi) \rvert \lvert \L\rvert + L^2\ln(L) \,\frac{\mu}{2\pi}\,\mathsf{I}(f) \, \mathcal H^2(\p\L) + o(L^2\ln(L))
\ee
as $L\to\infty$. To this end, note that $|\G| = \frac{4\pi}{3} \mu^{3/2}$ and that our functional $\mathsf{I}$ here is the same as the functional $I$ in \cite{LSS1}. The double-surface integral $J(\p\G,\p\L)$~\cite[(2)]{LSS1} equals $\frac{\mu}{2\pi}\mathcal H^{2}(\p\L)$.

Letting $B$ tend to zero in \eqref{main formula} but keeping the Fermi energy $\mu$ fixed, the prefactor $\nu B$ tends to $\mu/2$. The remaining integral over $\p\L$ is independent of the strength $B$ and remains fixed. For the volume term we have in this limit
\begin{align*} \frac{B}{2\pi^2}\sum_{\ell=0}^{\nu-1} \sqrt{\mu - (2\ell+1)B} &\sim \frac{\mu^{3/2}}{4\pi^2 \nu} \sum_{\ell=0}^\nu \sqrt{1-\ell/\nu} \sim \frac{\mu^{3/2}}{4\pi^2} \int_0^1\text dx\,\sqrt{x} = \frac{\mu^{3/2}}{6\pi^2}\,.
\end{align*}
In this limit the volume term equals the above volume term at $B=0$ as in \eqref{LSSformula}. To summarize, we obtain 
\begin{align*} \lim_{B\downarrow0} \mbox{rhs of } \eqref{main formula} &= L^3 f(1) \frac{\mu^{3/2}}{6\pi^2}|\L| + L^2\ln(L) \frac{\mu}{2\pi} \,\mathsf I(f) \int_{\p\L} \mathrm{d}\mathcal H^2(v) \,|n(v)\cdot e_3| + o(L^2\ln(L))\,,
\end{align*}
which is identical to the right-hand side (rhs) of \eqref{LSSformula} except for the prefactor depending on $\p\L$.

\item[(iii)] There is no 'level mixing' at the order in $L^2\ln(L)$ in the sense that each Landau level enters individually in the numerical coefficient. In \cite{LSSmagn}, we proved that level mixing occurs in the two-dimensional setting at the next-to-leading order, namely at the order $L$. We expect level mixing to occur in the present case at the order $L^2$. This is certainly possible to prove, say for a cylindrical region, but it requires a three-term expansion in the $x^\parallel$-coordinate and the by now proved two-term expansion in the $x^\perp$-coordinate \cite{LSSmagn}. The caveat for us to proceed with this question is that the mentioned three-term expansion has not been proved so far for the entropy function. This is an interesting open problem.

\item[(iv)] For \eqref{main formula} to hold we require only weak regularity of the boundary $\p\L$ like in the proof in \cite{LSS1} for the ideal Fermi gas. In contrast, the proof of the corresponding two-term asymptotics for the two-dimensional model in \cite{LSSmagn} required $\mathsf{C}^3$ smooth regions. This smoothness was a technical condition and may not be necessary. On the other hand and more importantly, only the leading contribution of the two-dimensional Hamiltonian enters and the extra logarithm stems from an expansion in the longitudinal direction, where weaker conditions suffice.
\end{enumerate}
\end{remarks}

\section{Proof of \texorpdfstring{\autoref{thm:polynomial}}{Theorem \ref{thm:polynomial}}}
We split the proof into two steps. The first one is the lemma below, which reduces the computation of the trace to an integral of the trace of the projection operator $\mathds{1}[(-\mathrm{i}\nabla^\parallel)^2 \le \mu]$ localized to the sets $L\Lambda_{x^\perp}\subset \R$ with respect to $x^\perp\in\R^2$. The second step starts from there, proves an asymptotic expansion of this trace, and finishes the proof of \autoref{thm:polynomial}.

\begin{defin}  \label{def L_x}
For any Borel set $E\subset \R^3$ and any $x^\perp\in\R^2$ we define $E_{x^\perp} \coloneqq \{x^\parallel\in\R : (x^\perp,x^\parallel)\in E\}$ to collect the third components of the intersection $E\cap (\{x^\perp\}\times\R)$. 
\end{defin}

\begin{lemma}\label{lem:reduction}
Let $m \in \mathbb N$ with $m \ge 2$. Then, under the same conditions as in \autoref{thm:polynomial}(i), there is a constant $C$ depending only on $B,m$ and $\mu$ such that 
\begin{align} \label{not main formula} 
\Big|\tr (\mathrm{D}_\mu({L\L}))^m - L^2\frac{B}{2 \pi}\sum_{\ell=0}^{\nu-1} \int_{\R^2} \mathrm{d} x^\perp\, \operatorname{tr} \left( \mathds{1}_{L \L_{x^\perp}}\mathds{1}[(-\mathrm{i}\nabla^\parallel)^2 \le \mu(\ell)]  \mathds{1}_{L \L_{x^\perp}} \right)^m\Big|\le C \mathcal K(\L) L^2\,,
\end{align}
where the $\L$ dependent constant $\mathcal K(\L)$ is defined in \autoref{tubular nbh small}; it is positive and finite for any piecewise Lipschitz region $\L$. Note that $L \L_{x^\perp} \coloneqq L \left(\L_{x^\perp} \right)$ is (in general) different from $(L \L)_{x ^\perp}$. 
\end{lemma} 

\begin{proof} We utilize the same changes of coordinates in the first two components (that is, for the planar $x_0^\perp$-coordinates) as in \cite{LSSmagn}. For the convenience of the reader we repeat all steps. 

As $\L$ is bounded, the operator $\mathds{1}_{L\L} \mathrm{D}_\mu$ is Hilbert--Schmidt and therefore $\mathrm{D}_\mu(L \L)$ is trace-class. We may write
\be \tr \mathrm{D}_\mu(L \L)^m = \int_{\R^3} \text d x_0\, \mathrm{D}_\mu(L \L)^m(x_0,x_0)
\ee
with integral kernel
\be \mathrm{D}_\mu(L \L)(x,y) = \sum_{\ell=0}^{\nu-1} p_\ell(x^\perp,y^\perp) k_{\mu(\ell)}(x^\parallel,y^\parallel)\,,\quad x=(x^\perp,x^\parallel)\,,y=(y^\perp,y^\parallel)\,, \ee
as in \eqref{int kernel Dmu}. Therefore, the trace is of the form
\begin{align*} \tr \mathrm{D}_\mu(L \L)^m 
= & \int_{L\L} \mathrm dx_0\sum_{\ell_1,\ldots,\ell_m=0}^{\nu-1} \int_{\R^{2(m-1)}} \text d x^\perp_1\cdots\text d x^\perp_{m-1}\, p_{\ell_1}(x^\perp_0,x^\perp_1) p_{\ell_2}(x^\perp_1,x^\perp_2)\cdots p_{\ell_m}(x^\perp_{m-1},x^\perp_0) 
\\
&\times\,\int_{\R^{m-1}} \text d x_1^\parallel\cdots\text d x_{m-1}^\parallel\, k_{\mu(\ell_1)}(x^\parallel_0-x_1^\parallel) \cdots k_{\mu(\ell_m)}(x_{m-1}^\parallel-x^\parallel_0)\, 1_{L\L}(x_1)\cdots 1_{L\L}(x_{m-1})\,.
\end{align*}
We begin by approximating $1_{L\L}(x_j)$ by $1_{L\L}(x_0^\perp,x_j^\parallel)$. We call the resulting approximate term $T(L\L)$. This means
\begin{align*} 
T(L \L) \coloneqq &\int_{L \L} \mathrm d x_0\sum_{\ell_1,\ldots,\ell_m=0}^{\nu-1} \int_{\R^{2(m-1)}} \text d x^\perp_1\cdots\text d x^\perp_{m-1}\, p_{\ell_1}(x^\perp_0,x^\perp_1) p_{\ell_2}(x^\perp_1,x^\perp_2)\cdots p_{\ell_m}(x^\perp_{m-1},x^\perp_0) 
\\
&\times\,\int_{\R^{m-1}} \text d x_1^\parallel\cdots\text d x_{m-1}^\parallel\, k_{\mu(\ell_1)}(x_0^\parallel-x_1^\parallel) \cdots k_{\mu(\ell_m)}(x_{m-1}^\parallel-x_0^\parallel)\,  \\
&  \quad \times1_{L\L}(x_0^\perp,x_1^\parallel)\cdots 1_{L\L}(x_0^\perp,x_{m-1}^\parallel)\, . 
\end{align*}
As the second line is independent of $x_j^\perp$, the integrals over $x_1^\perp, \ldots , x_{m-1}^\perp$ can be easily resolved and yield the diagonal of the integral kernel of the operator $\mathrm{P}_{\ell_1} \cdots \mathrm{P}_{\ell_{m}}$ at $x_0^\perp$, which is $B/(2\pi)$, if $\ell_1= \dots =\ell_m$ and $0$ otherwise. Thus, we have
\begin{align*} T(L \L) &= \int_{L \L} \mathrm dx_0\sum_{\ell=0}^{\nu-1} \frac B {2\pi} 
 \quad \int_{\R^{m-1}} \text d x_1^\parallel\cdots\text d x_{m-1}^\parallel\, k_{\mu(\ell)}(x_0^\parallel-x_1^\parallel)\cdots k_{\mu(\ell)}(x_{m-1}^\parallel-x_0^\parallel)\, \\
 & \qquad \times 1_{L\L}(x_0^\perp,x_1^\parallel)\cdots 1_{L\L}(x_0^\perp,x_{m-1}^\parallel)\,.  
\end{align*}
Now, we set $x^\perp \coloneqq x_0^\perp /L $ and observe $1_{L \L}(x^\perp_0, x_j^\parallel)=1_{L \L_{x ^\perp}}(x_j^\parallel)$. Thus, we have
\begin{align*} T(L \L) &= L^2 \int_{\R^2} \mathrm d x ^\perp  \sum_{\ell=0}^{\nu-1} \frac B {2\pi}  \int_{L \L_{x^\perp} } \mathrm dx_0^\parallel 
 \int_{\R^{m-1}} \text d x_1^\parallel\cdots\text d x_{m-1}^\parallel\, \\
 & \qquad \times k_{\mu(\ell)}(x_0^\parallel-x_1^\parallel) \cdots k_{\mu(\ell)}(x_{m-1}^\parallel-x_0^\parallel)\,  1_{L\L}(Lx^\perp,x_1^\parallel)\cdots 1_{L\L}(Lx^\perp,x_{m-1}^\parallel)\,  \\
&= L^2 \frac{B}{2\pi}\int_{\R^2} \mathrm d x ^\perp  \sum_{\ell=0}^{\nu-1} \tr \left(\mathds{1}_{L \L_{x ^\perp} } \mathds{1}[(-\mathrm{i}\nabla^\parallel)^2 \le \mu(\ell)]  \mathds{1}_{L \L_{x^\perp} } \right)^m\,,
\end{align*}
which is the expression in the claim. Thus, we are left to bound the error term of our approximation. Let us denote by $U \subset \R^{3m}$ the set of all tuples $(x_0,x_1, \dots ,x_{m-1})$ where $1_{L\L} (x_0) 1_{L \L} (x_1) \cdots 1_{L \L} (x_{m-1}) $ is not equal to $1_{L\L} (x_0) 1_{L \L} (x_0^\perp,x_1^\parallel) \cdots 1_{L \L} (x_0^\perp,x_{m-1}^\parallel)$. Then, using the notation $x_m \coloneqq x_ 0$ we trivially have
\be \big\lvert T(L \L) - \tr \mathrm{D}_\mu(L \L)^m  \big\rvert \le \int_U \mathrm dx_0 \mathrm dx_1 \cdots \mathrm dx_{m-1} \prod_{j=0}^{m-1} \lvert \mathrm D_\mu(x_j, x_{j+1} ) \rvert\,. \ee

We will now enlarge $U$ until we get a set where the integral can easily be calculated. Let $(x_0,x_1,\dots, x_{m-1}) \in U$. Then, there is a $j \in \{ 1, \dots , m-1\}$ such that $1_{L \L} (x_j) \neq 1_{L \L} (x_0^\perp, x_j^\parallel)$. Thus, the line between $x_j$ and $(x_0^\perp,x_j^\parallel)$ has to intersect the boundary $L \partial \L$, which implies $\operatorname{dist}(x_j, L \partial \L) \le \lVert x_j^\perp - x_0^\perp \rVert$. By the triangle and mean inequalities, we observe that
\be \operatorname{dist}(x_j, L \partial \L) \le \lVert x_j^\perp - x_0^\perp \rVert \le \sum_{k=1} ^{m} \lVert x_k^\perp - x_{k-1}^\perp  \rVert \le  \sqrt{m}  \sqrt { \sum_{k=1} ^{m} \lVert x_k^\perp - x_{k-1}^\perp  \rVert ^2  } \,.\label{Uj def approx}\ee
For $j \in \{0, \dots, m-1\}$, let $U_j \subset \R^{3m}$ be the set of all $(x_0,x_1, \dots, x_{m-1})\in\R^{3m}$ satisfying 
\be \operatorname{dist}(x_j, L \partial \L) \le {\sqrt {m} }  \sqrt { \sum_{k=1} ^{m} \lVert x_k^\perp - x_{k-1}^\perp  \rVert ^2  } \,.  \ee
As $U \subset \bigcup_{j=1}^{m-1} U_j$, we see that
\begin{align}
 \int_{U}  \mathrm dx_0 \mathrm dx_1 \cdots \mathrm dx_{m-1} \prod_{j=0}^{m-1} \lvert \mathrm D_\mu(x_j, x_{j+1} ) \rvert &\le \sum_{k=1} ^{m-1}  \int_{U_k} \mathrm dx_0 \mathrm dx_1 \cdots \mathrm dx_{m-1} \prod_{j=0}^{m-1} \lvert \mathrm D_\mu(x_j, x_{j+1} )\rvert  \\
& = (m-1)  \int_{U_0} \mathrm dx \mathrm dx_1 \cdots \mathrm dx_{m-1} \prod_{j=0}^{m-1} \lvert \mathrm D_\mu(x_j, x_{j+1} )\rvert\,.
\end{align}
The cyclic parameter shift $(x,x_1, \dots, x_{m-1}) \mapsto (x_1, x_2, \dots , x)$ sends $U_j$ to $U_{j+1}$ and does not change the integrand. For $1 \le j \le m$, let $y_j \coloneqq x_j- x_{j-1}$. We will change variables from $(x_0,x_1, \dots, x_{m-1})$ to $(x_0,y_1, \dots , y_{m-1}) \eqqcolon (x_0, \mathbf y)$. Using $y_m^\perp = - \sum_{j=1}^{m-1} y_j^\perp$, similar to \eqref{Uj def approx}, we observe that
\be
m \sum_{k=1} ^{m} \lVert x_k^\perp - x_{k-1}^\perp  \rVert ^2 = m( \lVert \mathbf y ^\perp \rVert^2 + \lVert y_m^\perp \rVert^2) \le m \lVert \mathbf y ^\perp \rVert^2  + m(m-1) \lVert \mathbf y ^\perp \rVert ^2= m^2 \lVert \mathbf y^\perp \rVert^2\,.
\ee
Thus, under this change of variables the set $U_0$ is mapped into the set  
\be V \coloneqq \left\{(x_0,y_1, \dots, y_{m-1} ) \in \R^{3m} \colon\operatorname{dist}(x_0, L \partial \L) \le m \lVert \mathbf y^\perp \rVert \right\} \,.\ee 
Let us first estimate the integrand in terms of the $y_j$'s. With \eqref{int kernel Dmu}, \eqref{Landau_kernel:eq} and \eqref{k mu}, we get  
\be \lvert \mathrm D_\mu(x_j, x_{j+1} ) \rvert \le  C_{\mu,B,1} \frac  {\exp(-B \lVert y_{j+1}^\perp \rVert^2/8)} {\langle y_{j+1}^\parallel \rangle} \,. \label{kernel est00} \ee 
We recall that $\langle x \rangle=\sqrt{1+ x^2}$ is the Japanese bracket. 

For $x_0 \in\R^3$, let $\Omega_{x_0} \coloneqq \{ \mathbf y ^\perp \in \R^{2(m-1)} \colon \operatorname{dist}(x_0, L \partial \L) \le m\lVert \mathbf y^\perp \rVert  \}$, and thus $V=\{(x_0,\mathbf y) \in \R^{3m} \colon \mathbf y^\perp \in \Omega_{x_0}\}$. We have
\begin{align}
\int_{U_0} \mathrm dx_0 \mathrm dx_1 &\cdots \mathrm dx_{m-1} \prod_{j=0}^{m-1} \lvert \mathrm D_\mu(x_j, x_{j+1} )\rvert \\
&\le C_{\mu,B,1}^m \int_{V} \mathrm dx_0 \mathrm d \mathbf y \prod_{j=1}^m \frac  {\exp(-B \lVert y_{j}^\perp \rVert^2/8)} {\langle y_{j}^\parallel \rangle} \\
&= C_{\mu,B,1}^m  \left(\int_{\R^{m-1}} \mathrm d\mathbf y^\parallel \prod_{j=1}^m  \frac  {1} {\langle y_{j}^\parallel \rangle}  \right)\int_{\R^3} \mathrm dx_0   \int_{\Omega_{x_0}} \mathrm d\mathbf y^\perp  \exp(- B  \lVert \mathbf y^\perp \rVert^2/8 ) \,. 
\end{align}
We need the estimate 
\be  \int_{\R^{m-1}} \mathrm d\mathbf y^\parallel \prod_{j=1}^m \frac  {1} {\langle y_{j}^\parallel \rangle} \le 2^m m! \label{kmu integrable eq}\,, \ee
 which is proved in \autoref{kmu integrable section}. We also have the bound
\begin{align}
 \int_{\Omega_{x_0}} \mathrm d\mathbf y^\perp  \exp(- B   \lVert \mathbf y^\perp \rVert^2 /8) &\le \sup_{\mathbf y^\perp \in \Omega_{x_0}} \left(\exp(- B \lVert \mathbf y^\perp \rVert^2 /9)\right) \int_{\R^{2(m-1)}} \mathrm d\mathbf y^\perp \exp(- B   \lVert \mathbf y^\perp \rVert^2 /72) \\
 &= \exp\left( \frac{ - B\operatorname{dist}(x_0, L \p\L)^2  } {9m^2 } \right) \sqrt{72 \pi/B }^{2(m-1)} \,.
 \end{align}
 Thus, we arrive at
 \begin{align}
 \int_{U_0} \mathrm dx \mathrm dx_1 \cdots \mathrm dx_{m-1} \prod_{j=0}^{m-1} \lvert \mathrm D_\mu(x_j, x_{j+1} )\rvert 
&\le  C_{\mu,B,1}^m C_{B,2}^m m! \int_{\R^3} \mathrm dx_0 \exp\left( \frac{ - B\operatorname{dist}(x_0, L \p\L)^2  } {9m^2 } \right) \\ 
&\le  C_{\mu,B,1}^m C_{B,2}^m  m! \sum_{k=0}^\infty \left \lvert B_{k+1} (L \partial \L) \right\rvert \exp \left(-\frac B {9 m^2} k^2 \right)\,.
\end{align} 
Here, we used an $(1,\infty)$ H\"older estimate on the sets $k \le \operatorname{dist} (x_0,L\partial \L)\le k+1$ for the integral over $\R^3$. We then enlarged these sets to the $k+1$-neighborhood $B_{k+1}(L \partial \L)$, as their measures can be estimated more easily. Thus, using \autoref{tubular nbh small} with $d=2$ and $r=k+1$, we arrive at
\begin{align}
&\lvert T(L \L) - \tr \mathrm{D}_\mu(L \L)^m  \rvert  \\
&\le (m-1)(C_{\mu,B,1} C_{B,2} )^m m! \sum_{k=0} ^\infty \left \lvert B_{k+1} (L \partial \L) \right\rvert \exp \left(-\frac B {9m^2} k^2 \right) \\
&\le (m-1)(C_{\mu,B,1} C_{B,2} )^m m! \sum_{k=0} ^\infty L^3 \left \lvert B_{\frac{k+1}L} ( \partial \L) \right\rvert \exp \left(-\frac B {9 m^2} k^2 \right) \\
&\le  (m-1)(C_{\mu,B,1} C_{B,2} )^m m! \sum_{k=0} ^\infty L^3  \mathcal K (\L)   \left( \frac{k+1} L + \frac{(k+1)^3} {L^3} \right) \exp \left(-\frac B {9 m^2} k^2 \right) \\
&\le  (m-1)(C_{\mu,B,1} C_{B,2} )^m m!  \mathcal K (\L )  L^2 \sup_{t>0}\left ((t+1)^3(t+2)^2 \exp\left(-\frac B {9m^2} t^2 \right) \right) \sum_{k=0} ^\infty \frac 1 {(k+1)(k+2)} \\
&\le  (m-1)(C_{\mu,B,1} C_{B,2} )^m m!  \mathcal K (\L)  L^2 C_{B,3} m^{5} \le \mathcal K (\L ) L^2 C_{\mu,B}^m m! \, ,
\end{align}
which was our claim.
\end{proof}

In the next step we accomplish the  
\begin{proof}[Proof of \autoref{thm:polynomial}] As the expression is linear in $f$, it suffices to consider monomials $f(t) = t^m$ with integer $m\ge1$. In the special case $m=1$, we just use \eqref{int kernel Dmu}, \eqref{Landau_kernel:eq}, and \eqref{k mu} to see
\begin{align} 
 \tr \mathrm D_{\mu}(L\L) &= \int_{L\L} \mathrm d x_0 \, \mathrm D_{\mu}(x_0,x_0) = \int_{L\L} \mathrm d x_0 \sum_{\ell=0}^{\nu-1} k_{\mu(\ell)}(0) p_\ell(x_0^\perp,x_0^\perp) \\
 &= \int_{L\L} \mathrm d x_0 \sum_{\ell=0}^{\nu-1} \frac{\sqrt{\mu(\ell)} }\pi  \frac B {2 \pi} = L^3 \frac B {2 \pi^2} \sum_{\ell=0}^{\nu-1} \sqrt{\mu(\ell)} 1^1 \lvert \L \rvert \,. 
 \end{align}
As $\mathsf I(1)=\mathsf I(id)=0$, this covers the case $m=1$ and we may from now on assume $m\ge 2$.

Our first aim is to understand the open sets $\L_{x^\perp}$. This is essentially a question about the nature of the sets $\L$. There are some results to choose from, so let us take a look. 
Due to \autoref{null set1} and \autoref{null sets2}, for Lebesgue almost every $x^\perp \in \R^2$, the set $\L_{x^\perp}$ is a finite union of disjoint intervals, $\p\left(\L_{x^\perp} \right)=(\p \L)_{x^\perp}$,  and $\#(\partial (\L_{x^\perp}))$ is twice the number of these intervals. Henceforth, we set $\p \L_{x^\perp} \coloneqq \p\left(\L_{x^\perp}\right)$. The (improved) asymptotic expansion
goes back to Landau and Widom\cite{LW} and is presented in \autoref{Appendix C}, see \autoref{1D final asym}. The coefficient $\mathsf{I}(m) = -1/(4\pi^2)\,\sum_{r=1}^{m-1}r^{-1}$ is mentioned below \eqref{def: I}.
 
For fixed $\L_{x^\perp}$, the error term $\varepsilon ( \L_{x^\perp}, L) $ remains bounded as $L \to \infty$. However, we need to know, whether this error term is integrable over $x^\perp$. Thus, the dependency on $\L_{x^\perp}$ is relevant. 
 
To derive the  $o(L^2 \ln(L))$ error term, we subtract the volume term, divide by $L^2 \ln(L)$ and use dominated convergence in order to exchange the limit $L \to \infty$ with the integral over $x^\perp$. Thus, instead of an estimate for the error term that is of a lower order in $L$ than $\ln(L)$, we only need an upper bound for the difference to the volume term, which is of order $\ln(L)$. This upper bound is provided by \autoref{lem: 1d bound}. As any interval in $L \L_{x^\perp}$ has length at most $CL$, we arrive at
\begin{align}
  \Big|\operatorname{tr}& \left(\mathds{1}_{L \L_{x^\perp} }\mathds{1}[(-\mathrm{i}\nabla^\parallel)^2 \le \mu(\ell)]  \mathds{1}_{L \L_{x^\perp} } \right)^m  -  \frac {\sqrt{ \mu(\ell)}} \pi L \lvert \L_{x^\perp} \rvert \Big|
\\
&= \Big \lvert \operatorname{tr}\Big[ \left( \mathds{1}_{L \L_{x^\perp} }\mathds{1}[(-\mathrm{i}\nabla^\parallel)^2 \le \mu(\ell)]  \mathds{1}_{L \L_{x^\perp} } \right)^m  - \mathds{1}_{L \L_{x^\perp} }\mathds{1}[(-\mathrm{i}\nabla^\parallel)^2 \le \mu(\ell)]  \mathds{1}_{L \L_{x^\perp} } \Big]\Big\rvert \\
& \le \left \lVert  \left( \mathds{1}_{L \L_{x^\perp} }\mathds{1}[(-\mathrm{i}\nabla^\parallel)^2 \le \mu(\ell)]   \mathds{1}_{L \L_{x^\perp} } \right)^m  -  \mathds{1}_{L \L_{x^\perp} }\mathds{1}[(-\mathrm{i}\nabla^\parallel)^2 \le \mu(\ell)]   \mathds{1}_{L \L_{x^\perp} } \right \rVert_1 \\
&\le  C \#( \partial \L_{x^\perp})  \ln(L) \,,
  \end{align}
where the constant $C$ depends on $m$, $\mu(\ell)$ and $\L$, but not on $x^\perp$. With this estimate, we apply dominated convergence to get
  \begin{align}
  &\lim_{L \to \infty} \frac 1 {L^2 \ln(L)} \Big(  \tr \mathrm{D}_\mu(L\L)^m  - BL^3 \lvert \L \rvert\sum_{\ell=0}^{\nu-1} \frac {\sqrt{\mu(\ell)}} {2 \pi^2} \Big) \\
  &= \sum_{\ell=0}^{\nu-1}  \lim_{L \to \infty} \frac 1 {L^2 \ln(L)}  \frac {B} {2\pi} L^2    \Big(\int_{\R^2} \mathrm{d}x^\perp \,\operatorname{tr} \big( \mathds{1}_{L \L_{x^\perp} }\mathds{1}[(-\mathrm{i}\nabla^\parallel)^2 \le \mu(\ell)]  \mathds{1}_{L \L_{x^\perp} } \big)^m  - \frac {\sqrt{\mu(\ell)}} \pi \lvert L \L_{x^\perp} \rvert \Big)
\\ 
  &= \sum_{\ell=0}^{\nu-1}   \frac {B} {2\pi} \int_{\mathbb R^2} \mathrm{d}x^\perp \lim_{L \to \infty} \frac 1 {\ln(L)} \Big(  \operatorname{tr} \big(\mathds{1}_{L \L_{x^\perp} }\mathds{1}[(-\mathrm{i}\nabla^\parallel)^2 \le \mu(\ell)]  \mathds{1}_{L \L_{x^\perp} } \big)^m  - \frac {\sqrt{\mu(\ell)}} \pi \lvert L \L_{x^\perp} \rvert \Big) 
\\
&=  \nu \frac{B}{2\pi} 2\,\mathsf{I}(m) \int_{\mathbb R^2} \mathrm{d}x^\perp \, \#( \partial \L_{x^\perp} )  = \nu B \,\mathsf{I}(m) \frac{1}{\pi}\int_{\p\L} \mathrm{d}\mathcal H^2(v) \, \lvert n(v)\cdot e_3 \rvert \,.
  \end{align}
We moved the sum over $\ell$ to the front, as every summand converges as $L \to \infty$. In the second line we used that $\int_{\R^2} \text d x^\perp \,|\L_{x^\perp}| = |\L|$. Finally we inserted \eqref{B.3: f=1} to obtain the expansion with error term $o(L^2\ln(L))$ as claimed in the theorem.

\medskip
For the second part, we need to show that the error term for polynomials can be bounded by $CL^2$, if $\L$ is a piecewise $\mathsf{C}^{1,\alpha}$ region for some $0< \alpha <1$, as defined in \autoref{def ps}. This time, we use \autoref{1D final asym} to deal with the trace of the one-dimensional operator. For that, we arrange each $\partial \L_{x^\perp}\coloneqq (\p \L)_{x^\perp}=\{w_{x^\perp1}, \dots, w_{x^\perp \#( \partial \L_{x^\perp})}\}\subset\R$ in the order of increasing third components and write
\begin{align} 
\Big|\operatorname{tr} &\left( \mathds{1}_{L \L_{x^\perp} }\mathds{1}[(-\mathrm{i}\nabla^\parallel)^2 \le \mu]  \mathds{1}_{L \L_{x^\perp} } \right)^m - \frac{\sqrt \mu}{\pi}  L \lvert \L_{x^\perp} \rvert - 2 \,\mathsf{I}(m) \#( \partial \L_{x^\perp} ) \ln(1+ L ) \Big| 
\\
&\le  C \sum_{i=1}^{\#( \partial \L_{x^\perp}) -1}\big(1+\lvert  \ln(\lvert w_{x^\perp i}-w_{x^\perp i+1} \rvert) \rvert\big) \\ 
&\le  C \sum_{i=1}^{\#( \partial \L_{x^\perp}) } \Big(1+ \lvert \ln\big( \inf_{v \in \partial \L_{x^\perp}\setminus w_{x^\perp i}} \lvert  w_{x^\perp i} -v\rvert \big)\rvert\Big)\,.
\end{align} 
In the last step, we used that the distance between any two points in $\p\L$ is bounded from above, as $\L$ is bounded to conclude that only short distances $\lvert v- v_i\rvert$ can lead to an error term larger than the $O(\#( \partial \L_{x^\perp}))$-term we have in front. A lower bound for the infimum is provided by \autoref{lx estimate}. This bound is zero in some cases, which leads to the logarithm being infinite. This just means that our integrand in the integral over $x^\perp$ attains infinity. The integral can still exist and we will show that it does.

As the terms of order $L^3$ and $L^2\ln(L)$ work just like in the previous case, we will only consider the error term. Hence, we need to estimate
\begin{align}
\int_{\R^2} \mathrm{d}x^\perp   & \sum_{i=1}^{ \#( \partial\L_{x^\perp})} \Big(1+ \lvert \ln( \inf_{v \in \partial \L_{x^\perp}\setminus w_{x^\perp i}} \lvert w_{x^\perp i} -v\rvert )\rvert \Big)
\\
&\le C \int_{\R^2} \mathrm{d}x^\perp   \sum_{i=1}^{\#(\partial \L_{x^\perp})} \Big(1+ \lvert \ln( \min\{ \operatorname{dist}(w_{x^\perp i},\Gamma), \lvert n((x^\perp, w_{x^\perp i})) \cdot e_3 \rvert^{\frac 1 \alpha})\rvert\}\Big) \\
&\le  C \int_{\R^2} \mathrm{d}x^\perp   \sum_{w \in  \{ x^\perp\} \times \partial \L_{x^\perp}} \Big(1+ \lvert \ln(\operatorname{dist}(w,\Gamma))\rvert + \lvert \ln(\lvert n(w) \cdot e_3 \rvert) \rvert \Big)\,.
\end{align}
In the first step, we applied \autoref{lx estimate} with the vectors $v_1 \coloneqq (x^\perp,w_{x^\perp i})$ amd $v_2 \coloneqq (x^\perp, v)$  noting that $\frac {v_1-v_2}{\|v_1 -v_2\|} = \pm e_3$. We now want to rewrite this integral as an integral over the boundary $\partial \L$. This is possible by \autoref{Lipschitz cov}. Hence, we have (recall that $\mathcal H^2$ is the canonical surface measure on $\p\L$),

\begin{align}
\int_{\R^2} \mathrm{d}x^\perp   & \sum_{i=1}^{\#( \partial\L_{x^\perp})} \Big(1+ \lvert \ln( \inf_{v \in \partial \L_{x^\perp}\setminus w_{x^\perp i}} \lvert w_{x^\perp i} -v\rvert )\rvert \Big)\\
&\le C \int_{\partial \L} \mathrm{d}\mathcal H^2(w)\, \big[1+ \lvert \ln(\operatorname{dist}(w,\Gamma))\rvert + \lvert \ln(\lvert n(w) \cdot e_3 \rvert) \rvert \big]\, \lvert n(w) \cdot e_3\rvert \\
&\le  C  +C \int_{\partial \L} \mathrm{d}\mathcal H^2(w)\,\lvert \ln(\operatorname{dist}(w,\Gamma))\rvert  \le C\,.
\end{align}
In the second step, we used that $0 \le \lvert n(w) \cdot e_3 \rvert \le 1$ and that for $0 \le t \le 1$, we have $ 0\le \lvert t \ln(t) \rvert \le 1/\mathrm{e}$. The last step is a rather lengthy, not particularly insightful calculation, which can be found in \autoref{dist to Gamma int lem}.

Once we put the factor $L^2$ back in front of this, we arrive at the error term $O(L^2)$ which completes the proof of the second part of this theorem.
\end{proof}

\section{Entanglement entropy}\label{Section:Entanglement}

Here is the main result of this paper.

\begin{thm}\label{entropy of ground state} Suppose that $\L\subset\R^3$ is a piecewise Lipschitz region and let $\mu> B$. Let $\nu\coloneqq\lceil \frac12(\mu/B -1)\rceil$ and let $h\colon [0,1] \to \R$ be a continuous function, which is $\b$-H\"older continuous at $0$ and $1$ for some $1 \ge \b>0$, and assume that $h(0)=h(1)=0$. Then, we have the asymptotic expansion
\be
\tr h(\mathrm D_\mu(L\L) )= L^2 \ln(L) \nu B \frac1 \pi \, \mathsf{I}(h) \int_{\p\L}\mathrm{d}\mathcal H^2(v)\,|n(v)\cdot e_3| + o(L^2\ln(L))  \,.
\ee
In particular, as the $\g$-R\'enyi entropy function $h_\g$ is $\b$-H\"older continuous for any $\b < \min(\g,1)$, the $\g$-R\'enyi entanglement entropy, $\mathrm S_\g(L\L)$, of the ground state at Fermi energy $\mu$ localized to $L\L$, satisfies the asymptotic expansion
\be 
\mathrm S_\g(L\L) = L^2\ln(L)\nu B\,\frac{1+\g}{24\g\pi}\int_{\p\L}\mathrm{d}\mathcal H^2(v)\,|n(v)\cdot e_3| + o(L^2\ln(L)) 
\ee
as $L\to\infty$. 
\end{thm}

We use certain estimates on traces. To this end, let us denote by $s_n({T}), n\in\N$, the singular values of the compact operator ${T}$ on a (separable) Hilbert space, arranged in decreasing order. The standard notation $\mathfrak S_p, 0<p<\infty$ is used for the class of operators with a finite \textit{Schatten--von Neumann quasi-norm}: 
\begin{align*}
\|{T}\|_p \coloneqq \bigg[\sum_{n=1}^\infty s_n({T})^p\bigg]^{\frac{1}{p}}<\infty\,.
\end{align*}
If $p\ge 1$, then $\| \cdot\|_p$ defines a norm. For $0 < p < 1$ it is a quasi-norm that satisfies the \textit{$p$-triangle inequality}
\begin{align}\label{p-tri:eq}
\| {T}_1+{T}_2\|_p^p\le \| {T}_2 \|_p^p + \|{T}_2\|_p^p\,.
\end{align}
The class $\mathfrak S_1$ is the standard trace-class. The class $\mathfrak S_2$ is the ideal of Hilbert--Schmidt operators. The $p$-Schatten quasi-norm estimate required for this proof is shown in \autoref{Hank:thm}.

\begin{proof}[Proof of \autoref{entropy of ground state}]
The proof goes along the same line of arguments as presented in \cite{LSS1} and \cite{LSSmagn}. We recall  that $\mathsf{I}(h_\g) = (1+\g)/(24\g)$ and thus we are left to show the claim for the function $h$. Let $r = \b/2$ and $\varepsilon>0$. We choose a smooth cutoff function  $\zeta_\varepsilon$ such that $0\le \zeta_\varepsilon\le 1$ and such that $\zeta$ vanishes on $[\varepsilon,1-\varepsilon]$ and equals $1$ on $[0,\varepsilon/2] \cup [1-\varepsilon/2,1]$. As $h$ is continuous and $\b$-H\"older continuous at $0$ and $1$, there is a constant $C$ such that
\be h(t) \le Ct^\b (1-t)^\b \,, \quad t \in [0,1]\,. \ee
This implies 
\begin{align}\label{zetah:eq}
|(\zeta_\varepsilon h)(t)| \le C \varepsilon^r t^r(1-t)^r\,\quad t \in [0,1]\,.
\end{align}
As the function $t \mapsto \frac {(1-\zeta_\varepsilon(t) )h(t) } {t(1-t)}$ is continuous, we can infer from the Stone--Weierstrass approximation theorem that there is a polynomial $p$ and a function $\delta_\varepsilon \colon [0,1] \to \R$ with $\lVert \delta_\varepsilon \rVert_{\Lp^\infty([0,1])} \le \varepsilon ^r $ and
\be \frac {(1-\zeta_\varepsilon(t)) h(t) } {t(1-t)} = p(t) + \delta_\varepsilon(t)\,, \quad  t \in [0,1] \,.\ee
Thus, we have
\begin{align}
h(t) = p(t) t(1-t) + \delta_\varepsilon(t) t(1-t)  + \zeta_\varepsilon(t) h(t) \eqqcolon p(t)t(1-t) + \phi_\varepsilon(t)\, .
\end{align}
As $t(1-t) \le t^r (1-t)^r$, we observe
\be \lvert \phi_\varepsilon(t) \rvert \le C\varepsilon^r t^r(1-t)^r\,, \quad  t \in [0,1] \,. \label{phi eps est} \ee
Thus, using \autoref{Hank:thm}, \eqref{Dmu exp} and  \eqref{p-tri:eq}, we arrive at
\begin{align}
 \lvert \tr \phi_\varepsilon(\mathrm D_\mu(L \L) ) \rvert &\le   C\varepsilon^r\tr ( \mathrm D_\mu(L\L)^r(1-\mathrm D_\mu(L\L))^r ) \\
 &=C\varepsilon^r \lVert \mathds {1}_{L \L} \mathrm D_\mu  \mathds {1}_{L \L^\complement} \mathrm D_\mu \mathds {1}_{L \L} \rVert_r^r \label{4.10} \\
 &=C\varepsilon^r \lVert  \mathds {1}_{L \L} \mathrm D_\mu  \mathds {1}_{L \L^\complement}  \rVert_{2r}^{2r} \\
 &\le C \varepsilon^r \sum_{\ell=0}^{\nu-1} \lVert  \mathds {1}_{L \L} (\mathrm{P}_\ell \otimes \mathds1( (-\mathrm{i} \nabla^\parallel)^2 \le \mu(\ell)))  \mathds {1}_{L \L^\complement}  \rVert_{2r}^{2r} \\
 &\le C\varepsilon^r C L^2 \ln(L)\,.  \label{phi eps trace est}
 \end{align}
 In \eqref{4.10}, we used that $\mathrm D_\mu$ is a projection. 
Let $q(t) \coloneqq p(t)t(1-t)$. Now, by linearity of $\mathsf {I}$ and the estimate \eqref{phi eps est}, we have
\be \lvert \mathsf{I}(h) - \mathsf{I}(q) \rvert = \lvert \mathsf{I} (\phi_\varepsilon) \rvert \le C \varepsilon^r \,\mathsf{I} (t \mapsto t^r(1-t)^r) \le C \varepsilon^r \,.\label{phi eps I est}\ee
\autoref{thm:polynomial}(i) applied for the polynomial $q$ with $q(0)=q(1)=0$ yields
\be \tr q(\mathrm D_\mu(L\L)) = L^2 \ln(L) B \nu \,\mathsf{I} (q) \frac 1 \pi\int_{\p\L}\mathrm{d}\mathcal H^2(v)\,|n(v)\cdot e_3| + o(L^2\ln(L)) \label{q asym exp}\,.  \ee
Now, combining \eqref{phi eps trace est}, \eqref{phi eps I est} and \eqref{q asym exp}, we arrive at
\be \limsup_{L\to\infty} \left|\frac{\tr h(\mathrm{D}_\mu(L\L))}{L^2\ln(L)} - \nu B \,{\sf I}(h)\,\frac{1}{\pi} \int_{\p\L} \mathrm{d} \mathcal H^2(v) \,|n(v)\cdot e_3| \right|\le C \varepsilon^r\,.
\ee
As $\varepsilon>0$ is arbitrary, we have proved the claim.
\end{proof}

\section{Schatten--von Neumann quasi-norm estimates}\label{Section: Schatten}

By a box in $\mathbb R^d$ we mean a Cartesian product of $d$ intervals. These intervals do not have to be bounded. We will denote subsets of $\mathbb R$ by $I$, of $\mathbb R^2$ by $\Upsilon$, and of $\mathbb R^3$ by $\Lambda$. We will combine known estimates for the two-dimensional magnetic Hamiltonian from \cite{LSSmagn} and for the one-dimensional Hamiltonian \cite{LSS1,Sob:Schatten} without a magnetic field.

Let $\Upsilon,\Upsilon'\subset \mathbb R^2$ be Lipschitz regions and let $I,I' \subset \mathbb R$ be finite unions of closed intervals. 
 Then we have
\begin{align}
\mathds{1}_{\Upsilon \times I}\big( \mathrm{P}_{ \ell } \otimes  \mathds{1}[(-\mathrm{i}\nabla^\parallel)^2 \le \mu]\big) \mathds{1}_{\Upsilon'\times I'} = \big(\mathds{1}_{\Upsilon} \mathrm{P}_{ \ell } \mathds{1}_{\Upsilon'} \big) \otimes \big(\mathds{1}_{I}  \mathds{1}[(-\mathrm{i}\nabla^\parallel)^2 \le \mu] \mathds{1}_{I'} \big) \label{product1}\,.
\end{align}
As the singular values of the tensor product of two operators are given by all possible products of pairs of the individual singular values, we have for any $0<p \le \infty \colon$
\begin{align}
\left \lVert  \mathds{1}_{\Upsilon \times I}\left( \mathrm{P}_{ \ell } \otimes  \mathds{1}[(-\mathrm{i}\nabla^\parallel)^2 \le \mu]\right) \mathds{1}_{\Upsilon'\times I'}\right \rVert_p = \left \lVert \mathds{1}_{\Upsilon} \mathrm{P}_{ \ell } \mathds{1}_{\Upsilon'} \right\rVert_p \left \lVert \mathds{1}_{I}  \mathds{1}[(-\mathrm{i}\nabla^\parallel)^2 \le \mu] \mathds{1}_{I'} \right\rVert_p \label{productnorm1}\,.
\end{align}
The following general properties will be useful:

\begin{lemma} \label{localization properties}
For any self-adjoint bounded operators $S,T\colon \mathsf{L}^2(\mathbb R^d) \to \mathsf{L}^2(\mathbb R^d)$, any measurable sets $\Omega_1,\Omega_2$, $\Omega_1',\Omega_2' \subset \mathbb R^d$ and any $0<p \le 1$, we have 
\begin{description}
\item[Symmetry] $ \left \lVert \mathds{1}_{\Omega_1} T \mathds{1}_{\Omega_2} \right \rVert_p = \left \lVert \mathds{1}_{\Omega_2} T \mathds{1}_{\Omega_1} \right \rVert_p$,
\item[Monotonicity I] $ \left \lVert \mathds{1}_{\Omega_1} T \mathds{1}_{\Omega_2} \right \rVert_p \le  \left \lVert \mathds{1}_{\Omega_1 \cup \Omega_1'} T \mathds{1}_{\Omega_2\cup \Omega_2'} \right \rVert_p$,
\item[Monotonicity II] If $0 \le S \le T$, then  $\lVert S \rVert_p \le \lVert T \rVert_p$,
\item[Subadditivity] $\left \lVert \mathds{1}_{\Omega_1 \cup \Omega_1'} T \mathds{1}_{\Omega_2} \right \rVert_p^p \le  \left \lVert \mathds{1}_{\Omega_1} T \mathds{1}_{\Omega_2} \right \rVert_p^p + \left \lVert \mathds{1}_{\Omega_1'} T \mathds{1}_{\Omega_2} \right \rVert_p^p$.
\end{description}
\end{lemma}
A proof of these properties can be found for example in \cite{P2021}.

We assume now that the magnetic-field strength has been ``scaled out" so that $B=1$ for the remainder of this section. The effective scale in the planar coordinates is $L\sqrt{B}$ and in the perpendicular it is $L\sqrt{\mu}$.

Next, we collect some more specific (quasi-)norm estimates for both the one dimensional free Hamiltonian and the constant magnetic field Hamiltonian in two dimensions. 
\begin{prop} \label{input prop}
Let $0< p \le 1, \ell \in \mathbb N_0$ and let $\mu > 0$. Then there is a constant $C$ such that for any $x \in \mathbb R^2, t \in \mathbb R, h\ge 2, \d \ge 1 $, any measurable set $\Upsilon \subset \R^2 $ such that $[-\d,1+\d]^2+x \subset \Upsilon$ and any measurable set $I\subset \R$ such that $[t,t+h] \subset I$, we have the estimates 
\begin{align}
\left \lVert\mathds{1}_{[0,1]^2+x} \mathrm{P}_{ \ell } \right \rVert_p^p &\le C\,,\\
\left \lVert\mathds{1}_{[0,1]^2+x} \mathrm{P}_{ \ell } \mathds{1}_{\Upsilon^\complement} \right \rVert_p^p &\le  C \exp( - p\d^2 /{18})\,, \label{5.4}\\
\left \lVert\mathds{1}_{[t,t+h]}  \mathds{1}[(-\mathrm{i}\nabla^\parallel)^2 \le \mu] \right \rVert_p^p &\le  C h \label{1Dvolume}\,,\\
\left \lVert\mathds{1}_{[t,t+h]}  \mathds{1}[(-\mathrm{i}\nabla^\parallel)^2 \le \mu] \mathds{1}_{I^\complement } \right \rVert_p^p &\le  C\ln (h) \label{1Dsurface}\,.
\end{align}
\end{prop}

\begin{proof}
The first two inequalities follow by \cite[Lemma 12]{LSSmagn}, monotonicity I in \autoref{localization properties}, and the unitary translation invariance of $\mathrm{P}_\ell$. The $8$ in the denominator was increased to $18$ in \eqref{5.4} as we switched from circles to squares\footnote{The choice of the value $18=3^2\times 2$ is convenient for this paper.}. To prove the last inequality, we first use monotonicity I and the translation invariance, then the standard unitary equivalence, see for example \cite[(7--10)]{LW}, and finally \cite[Corollary 4.7]{Sob:Schatten}. Thus
\begin{align}
\left \lVert\mathds{1}_{[t,t+h]}  \mathds{1}[(-\mathrm{i}\nabla^\parallel)^2 \le \mu] \mathds{1}_{I^\complement } \right \rVert_p^p 
&\le \left \lVert\mathds{1}_{[0,h]}  \mathds{1}[(-\mathrm{i}\nabla^\parallel)^2 \le \mu] \mathds{1}_{[0,h]^\complement } \right \rVert_p^p \\
&=  \left \lVert\mathds{1}_{[0,1]}  \mathds{1}[(-\mathrm{i}\nabla^\parallel)^2 \le h^2 \mu] \mathds{1}_{[0,1]^\complement } \right \rVert_p^p \\
&\le  C \ln(h)\,.
\end{align}
For the third inequality, we will reduce to the case $h=2, t=0$ by subadditivity, monotonicity I and translation invariance. Let $m \coloneqq \lceil h/2 \rceil \in\N$ be the smallest integer larger or equal to $h/2$. Thus, as $h\ge 2$, we have $m \le h$. We observe that
\begin{align}
\left \lVert\mathds{1}_{[t,t+h]}  \mathds{1}[(-\mathrm{i}\nabla^\parallel)^2 \le \mu] \right \rVert_p^p \le & \sum_{k=0}^{m-1} \left \lVert\mathds{1}_{[t+2k,t+2k+2]}  \mathds{1}[(-\mathrm{i}\nabla^\parallel)^2 \le \mu] \right \rVert_p^p \\
&= m \left \lVert\mathds{1}_{[0,2]}  \mathds{1}[(-\mathrm{i}\nabla^\parallel)^2 \le \mu] \right \rVert_p^p\\
&\le 2h \left \lVert\mathds{1}_{[0,2]}  \mathds{1}[(-\mathrm{i}\nabla^\parallel)^2 \le \mu] \right \rVert_p^p\,.
\end{align}
Using subadditivity once more we now estimate 
\begin{align}
\Big\lVert\mathds{1}_{[0,2]}  &\mathds{1}[(-\mathrm{i}\nabla^\parallel)^2 \le \mu] \Big\rVert_p^p \\
&\le  \left \lVert\mathds{1}_{[0,2]}  \mathds{1}[(-\mathrm{i}\nabla^\parallel)^2 \le \mu] \mathds{1}_{[0,2]} \right \rVert_p^p  + \left \lVert\mathds{1}_{[0,2]}  \mathds{1}[(-\mathrm{i}\nabla^\parallel)^2 \le \mu] \mathds{1}_{[0,2]^\complement}\right \rVert_p^p \\
&\le   \left \lVert\mathds{1}_{[0,2]}  \mathds{1}[(-\mathrm{i}\nabla^\parallel)^2 \le \mu] \mathds{1}_{[0,2]} \right \rVert_p^p  +  C \\
& = \left \lVert\mathds{1}_{[0,2]}  \mathds{1}[(-\mathrm{i}\nabla^\parallel)^2 \le \mu] \right \rVert_{2p}^{2p} +C  \,.
\end{align} 
The second summand was bounded by \eqref{1Dsurface} and the last quasi-norm identity is derived by the singular value identity $s_n(A)^2=s_n(A^*A)$. Define $Q \coloneqq \mathds{1}_{[0,2]}  \mathds{1}[(-\mathrm{i}\nabla^\parallel)^2 \le \mu ]$. Our claim is $Q \in \mathfrak S_p$ for all $0<p \le 1$. The last estimate shows that $Q \in \mathfrak S_p$, if  $Q \in \mathfrak S_{2p}$ for $p \le 1$. We now observe
\be \lVert Q \rVert_2^2= \int_0^1 \mathrm{d}s \int_0^1 \mathrm{d} t \, k_{\mu}^2(s-t) < \frac {\mu}{\pi^2}\,. \ee
Thus, we have $Q \in \mathfrak S_2$ and hence $Q \in \mathfrak S_{2^{1-n}}$ for any $n \in \mathbb N$. Lastly, as $\mathfrak S_p \subset \mathfrak S_q$ for $p<q$, we arrive at $Q \in \mathfrak S_p$ for any $0<p \le \infty$, which finishes the proof.
\end{proof}

After all these preparations we finally state the crucial local estimates that are needed in the proof of \autoref{Hank:thm}. 

\begin{lemma}\label{box3}
Let $0< p \le 1$. Then there is a constant $C$, such that for any $x \in \mathbb R^2, t \in \R, h\ge 2$, $\delta\ge 1$, any measurable $\Upsilon \subset \R^2$ such that $[-\delta,1+\delta]^2+x \subset \Upsilon \subset$ and any interval $I\subset \R$ such that $[t,t+h] \subset I \subset $, we have the estimates 
\begin{align}
\left \lVert \mathds{1}_{([0,1]^2+x) \times [t,t+h]}\left( \mathrm{P}_{ \ell } \otimes  \mathds{1}[(-\mathrm{i}\nabla^\parallel)^2 \le \mu]\right)\right \rVert_p^p  &\le C h\,,  \label{box3volume}\\
\left \lVert \mathds{1}_{([0,1]^2+x) \times [t,t+h]}\left( \mathrm{P}_{ \ell } \otimes  \mathds{1}[(-\mathrm{i}\nabla^\parallel)^2 \le \mu]\right) \mathds{1}_{\left( \Upsilon \times I\right)^\complement }  \right \rVert_p^p 
&\le  C h \exp \left( - p\delta^2 /{18} \right) + C \ln(h)  \label{box3surface}\,.
\end{align}
\end{lemma}

\begin{proof} 
For the first inequality, we use \eqref{productnorm1} and \autoref{input prop}. For the second inequality, we first observe
\begin{align}
(\Upsilon\times  I)^\complement = (\Upsilon^\complement\times  I) {\cup} ( \mathbb R^2\times  I^\complement) \subset (\Upsilon^\complement \times \R) \cup (\R^2 \times I^\complement)
\end{align}
and then we use the $p$-triangle inequality,  \eqref{productnorm1} and \autoref{input prop}. 
\end{proof}

We now fix a region $\Lambda \subset \mathbb R^3$ and define the signed distance function 
\begin{align} \label{def signed dist}
\operatorname{d}_\L(x) \coloneqq \begin{cases} +\operatorname{dist} (x, \partial \Lambda )  &\text{ for } x \not \in \Lambda \\
  -\operatorname{dist} (x, \partial \Lambda )  &\text{ for }  x  \in \Lambda   \end{cases}  \,,
\end{align}
where $\operatorname{dist}$ is the Euclidean distance. The signed distance function is Lipschitz-continuous with Lipschitz constant $1$.

In order to utilize \autoref{box3}, we need to essentially cover $L \Lambda$ with a lot of very long boxes (of dimensions $1 \times 1 \times O(L)$). This boils down to choosing appropriate intervals that cover most of $\Lambda_x$ (as defined in \autoref{def L_x}), for any $x \in \R^2$. Let  $G(x,\varepsilon)$ be the number of these intervals. The following lemma explicitly constructs  such intervals and lists the properties that $G(x,\varepsilon)$ and the intervals satisfy, which we need for our estimates. The basic idea is to collect  connected components of $\L_x$, which go sufficiently deep inside $\L$.
\begin{lemma} \label{Def G} 
For any $x \in \R^2$ and $\varepsilon >0$, there is a finite (possibly empty) set of intervals $A(x,\varepsilon)=\{ I_{1,x,\varepsilon}, \dots , I_{G(x, \varepsilon)),x,\varepsilon} \}$, satisfying the following conditions:
\begin{enumerate}
\item{ We have $\operatorname{d}_\L(I_{k,x,\varepsilon}) \subset (-\infty, -\varepsilon)$ and $\operatorname{dist}(I_{k,x,\varepsilon}, \partial \L)= \varepsilon$.}
\item{ For any $\lambda \in \Lambda_x$, there exists a $j$ with $1 \le j \le G(x,\varepsilon) \colon \lambda \in I_{j,x,\varepsilon}$, or $\operatorname{d}_\L((x,\lambda))>-2\varepsilon$.} 
\item{ We have $G(x, \varepsilon )= \# A(x,\varepsilon) \le \mathcal H^1\left(  \operatorname{d}_\L^{-1}((-2 \varepsilon, -\varepsilon)) \cap (\{x \} \times \R )\right)/\varepsilon$. }
\end{enumerate}
The signed distance function $\operatorname{d}_\L$, dependent on the piecewise Lipschitz region $\L$, is defined in \eqref{def signed dist}.
\end{lemma} 
We regard the lemma and its proof as the definitions of $A(x,\varepsilon),I_{j,x,\varepsilon}$ and $G(x,\varepsilon)$.
\begin{proof}
We consider the set $A_0(x,\varepsilon)$ of all connected components of $\left(\operatorname{d}_\L^{-1}((-\infty,-\varepsilon)) \right)_x \subset \R$ (with the convention that the empty set has no connected components). The set $A(x,\varepsilon)$ is defined as the set of all $I \in A_0(x,\varepsilon)$, such that there is a $\lambda \in I$ with $\operatorname{d}_\L((x,\lambda)) \le -2 \varepsilon$. The first point is already satisfied for all $ I \in A_0(x,\varepsilon)$ and thus holds for all $I$ in the smaller set $A(x,\varepsilon)$. For the second claim, we observe that if $\lambda \in \L_x$ with $\operatorname{d}_\L((x,\lambda)) \le -2 \varepsilon$, then $\lambda \in  \left(\operatorname{d}_\L^{-1}((-\infty,-\varepsilon)) \right)_x $ and thus there is an $I \in A_0(x,\varepsilon)$ with $\lambda \in I$. By definition of $A(x,\varepsilon)$, this ensures $I \in A(x,\varepsilon)$.

For the third claim, if $A(x, \varepsilon) \neq \emptyset$,  let $I =(\lambda_1, \lambda_4) \in A(x, \varepsilon)$ and define $\lambda_2 \coloneqq \inf \{ \lambda \in I \colon \operatorname{d}_\L( (x, \lambda) \le -2 \varepsilon \}, \lambda_3 \coloneqq \sup \{ \lambda \in I \colon \operatorname{d}_\L( (x, \lambda) \le -2 \varepsilon \}$. Thus, $\lambda_1 < \lambda_2<\lambda_3<\lambda_4$,  
\be \operatorname{d}_\L(\{x \} \times (\lambda_1,\lambda_2)) = \operatorname{d}_\L(\{x \} \times (\lambda_3,\lambda_4)) =(-2\varepsilon, -\varepsilon) \,,\ee
and, as $\operatorname{d}_\L$ has Lipschitz constant $1$, this means that 
 \be \mathcal H^1(  ( \{x \} \times I) \cap \operatorname{d}_\L^{-1}((-2 \varepsilon,-\varepsilon)) ) \ge \mathcal H^1(  \{x \} \times ( (\lambda_1,\lambda_2) \cup (\lambda_3,\lambda_4) ) ) \ge 2 \varepsilon\,. \ee
As $(\lambda_1, \lambda_2) \subset I$, $(\lambda_3, \lambda_4) \subset I$ and different elements of $A(x,\varepsilon)$ are disjoint (as they are connected components), we can sum the inequality over all elements of $A(x, \varepsilon)$ and arrive at 
 \be \mathcal H^1(  (\{x  \} \times \R )\cap  \operatorname{d}_\L^{-1} ((-2 \varepsilon, -\varepsilon))) \ge \sum_{I \in A(x, \varepsilon)}  \mathcal H^1(  (\{x \} \times I) \cap \operatorname{d}_\L^{-1}((-2 \varepsilon\,,\varepsilon)) )  \ge 2G(x,- \varepsilon ) \varepsilon\,, \ee
which implies the last claim (with an additional factor $1/2$).
\end{proof}

\begin{thm} Let $\L$ be a piecewise Lipschitz region (see \autoref{def ps}) and let $0<p \le1, \ell\in \N, \mu \in \R^+$. Then there are constants $L_0=L_0(\Lambda,p,\ell,\nu ) > 3$ and $C=C(\Lambda ,p,\ell,\nu)$ such that for all $L>L_0$ \label{Hank:thm}
\begin{align}
\left \lVert \mathds{1}_{L\Lambda } \mathrm{P}_{ \ell } \otimes  \mathds{1}[(-\mathrm{i}\nabla^\parallel)^2 \le \mu] \mathds{1}_{L\Lambda ^\complement} \right \rVert_p^p \le  C(\Lambda ,p) L^2 \ln(L)\,.
\end{align}
\end{thm}

\begin{proof}
We want to cover most of $L\Lambda $ with translates of cubes $[0,1]^2\times [0,h]$, where $h$ grows like $L$ and will use \autoref{box3} on these\footnote{The choice of $18$ in \autoref{input prop} makes the definition of $\delta$ quite nice to work in the estimate \eqref{interior cube estimate}.}. We set $\delta \coloneqq 6 p^{-1/2} \sqrt{ \ln (L)}$. Let $L_0$ be large enough to ensure that $\delta\ge 1 $ for $L>L_0$. Hence these cubes need to keep a distance of at least $\delta$ from the boundary. Set $\varepsilon \coloneqq \frac { 2(\delta+1)} L$. We also define the shorthand 
\begin{align}
\mathrm{P} \coloneqq  \mathrm{P}_{ \ell } \otimes  \mathds{1}[(-\mathrm{i}\nabla^\parallel)^2 \le \mu]\,.
\end{align}
Let $h_0$ be the length of the longest straight line contained in $\Lambda $.

Consider any $x \in \mathbb R^2$ with $G(x, \varepsilon)\ge1$, as defined in \autoref{Def G}. 
For $k \in \{1, \dots, G(x,\varepsilon)\}$, we define the boxes
\begin{align}
Q_{x,k}&\coloneqq  ([0,1]^2+Lx) \times (LI_{k,x,\varepsilon})& \subset L\Lambda\,,  \\
Q'_{x,k}&\coloneqq  ([-\delta,1+\delta]^2+Lx ) \times (LI_{k,x,\varepsilon} ) & \subset L\Lambda \,.
\end{align}
These inclusions hold because $\sqrt 2 < \sqrt 2 (\delta +1) <L \varepsilon =L \operatorname{dist} ( \{  x \} \times I_{k,x,\varepsilon} ,  \partial \L)  = \operatorname{dist} ( \{ L x \} \times(L I_{k,x,\varepsilon}) , L \partial \L)$.  

We assume $L>h_0$, $Lh_0>1$ and $L>2$. Now we have by monotonicity I in \autoref{localization properties} and by   \autoref{box3}
\begin{align}
\left \lVert \mathds{1}_{Q_{x,k}} \mathrm{P} \mathds{1}_{L\Lambda^\complement} \right \rVert _p^p & \le  \left \lVert \mathds{1}_{Q_{x,k}} \mathrm{P} \mathds{1}_{(Q'_{x,k})^\complement} \right \rVert _p^p \\
&\le   C L |I_{k,x,\varepsilon}| \exp \left( -p \delta^2 /18 \right) + C \ln(L|I_{k,x,\varepsilon}| ) \\
&\le  C h_0 L^{-1} +C \ln(L)+C \ln(h_0) \le C \ln(L)\,. \label{interior cube estimate}
\end{align}
The constant $C$ depends only on $p$ and $h_0$.

Now we consider some offset parameter $s \in [0,1)^2$, and we define 
\begin{align}
\Lambda _{\varepsilon,s} \coloneqq \Lambda  \setminus \bigcup _{z \in \mathbb Z^2} \bigcup_{k =1}^{G\left( \frac{ z+s}L ,\varepsilon \right)} \frac 1 LQ_{\frac{z+s} L ,k} \subset  \operatorname{d}_\L^{-1}\left( \left(  - 3\varepsilon ,0 \right) \right) \label{covering}\,.
\end{align}
The inclusion is based on the fact that for each $y \in \R^3$, there is a $\frac{z+s}L \in \mathbb R^2$ with $z\in\Z^2$ such that $y \in \left( \frac{z+s}L + \frac 1 L[0,1]^2 \right) \times \mathbb R$. If $\operatorname{d}_\L(y) \le -3 \varepsilon$, then the point $\left( \frac{z+s}L , y_3 \right)$ is at most $\frac{\sqrt 2}L < \varepsilon$ away from $y$ and hence at least $2\varepsilon$ away from the boundary. Therefore, there is a $k$ such that $y \in \frac 1 LQ_{\frac{z+s}L,k}$. 

We further define 
\begin{align}
Z_{\varepsilon} \coloneqq \left  \{u \in \mathbb Z^3 \colon \left(u+[0,1]^3\right) \cap L\operatorname{d}_\L^{-1}\left( \left(  - 3\varepsilon ,0 \right) \right) \not = \emptyset \right\}\,,
\end{align}
so that $L\operatorname{d}_\L^{-1}\left( \left(  - 3\varepsilon ,0 \right) \right) \subset \bigcup_{{u} \in Z_\varepsilon} \left({u}+ [0,1]^3\right)$. As $\varepsilon>\frac{\sqrt 3}L$, the length of the diagonal in a cube $\frac 1 L [0,1]^3$, we have (second inclusion)
\begin{align}
L\operatorname{d}_\L^{-1}\left( \left(  - 3\varepsilon ,0 \right) \right) \subset \bigcup_{u \in Z_\varepsilon} \left({u}+ [0,1]^3\right) \subset L\operatorname{d}_\L^{-1}\left( \left(  - 4\varepsilon , \varepsilon\right) \right)\,. \label{covering boundary}
\end{align}
Hence the volume of the middle term, which is the cardinality, $\# Z_\varepsilon$, of $Z_\varepsilon$, can be bounded by the volume of the right-hand side. For $\varepsilon<1$, using \autoref{tubular nbh small}, this is bounded by $L^3C(\Lambda )  \varepsilon$. Hence for $L > L_0$: 
\begin{align}
\# Z_\varepsilon \le L^3 C(\Lambda ) \varepsilon \le C(\Lambda ,p) L^2 \sqrt{ \ln (L)} \,.\label{Ze estimate}
\end{align}
Using the monotonicity I and subadditivity properties in \autoref{localization properties} and the covering \eqref{covering}, we can finally estimate,
\begin{align}
\left \lVert \mathds{1}_{L\Lambda } \mathrm{P} \mathds{1}_{L\Lambda^\complement} \right \rVert_p^p \le  \sum_{z \in \Z^2} \sum_{k=1}^{G\left( \frac{ z+s}L ,\varepsilon \right)} \left \lVert \mathds{1}_{Q_{\frac{z+s}L ,k} } \mathrm{P} \mathds{1}_{L\Lambda ^\complement} \right \rVert_p^p +\left \lVert \mathds{1}_{L\Lambda _{\varepsilon,s}} \mathrm{P} \mathds{1}_{L\Lambda ^\complement}\right \rVert_p^p\,.
\end{align}

The summands in the first sum can be bounded by $C\ln(L)$ using  \eqref{interior cube estimate}. The second term will be bounded using monotonicity I, \eqref{covering} and  \eqref{covering boundary} in the first step and using monotonicity II, subadditivity, \eqref{box3volume} and \eqref{Ze estimate} in the second step. Hence, we have
\begin{align}
\left \lVert \mathds{1}_{L\Lambda } \mathrm{P} \mathds{1}_{L\Lambda^\complement} \right \rVert_p^p &\le  \sum_{z \in \Z^2}G\left( \frac{ z+s}L ,\varepsilon \right)C \ln(L) + \sum_{{u} \in Z_\varepsilon} \left \lVert \mathds{1}_{[0,1]^3+{u}} \mathrm{P} \mathds{1}_{L\Lambda^\complement} \right \rVert_p^p \\
&\le   C \sum_{z \in \Z^2}G\left( \frac{ z+s}L ,\varepsilon \right) \ln(L)  + C(\Lambda ,p) L^2 \sqrt{ \ln(L)} {C} \,. 
\end{align}
For any fixed $L>L_0$ and $s \in [0,1)^2$, this is finite. Hence, we can integrate this over $s \in [0,1)^2$ and get a different upper bound. As the volume of $[0,1)^2$ is $1$, the left-hand side and the last term do not change, as it is an integral over a constant in both cases.
\begin{align}
\left \lVert \mathds{1}_{L\Lambda } \mathrm{P} \mathds{1}_{L\Lambda^\complement} \right \rVert_p^p \le C \int_{[0,1)^2} \mathrm{d}s \sum_{z \in \Z^2}G\left( \frac{ z+s}L ,\varepsilon \right) \ln(L)  + C(\Lambda ,p) L^2 \sqrt{ \ln(L)} \,.
\end{align}
Now we can use Fubini on the product $\mathbb Z^2 \times [0,1)^2= \mathbb R^2$. Hence we have
\begin{align}
\left \lVert \mathds{1}_{L\Lambda } \mathrm{P} \mathds{1}_{L\Lambda^\complement} \right \rVert_p^p &\le C \ln(L) \int_{\mathbb R^2} G\left( \frac xL ,\varepsilon \right) \,\mathrm{d}x  + C(\Lambda ,p) L^2 \sqrt{ \ln(L)} \\
&=  C \ln(L) L^2 \int_{\mathbb R^2} G(x,\varepsilon) \,\mathrm{d}x +C(\Lambda ,p) L^2 \sqrt{ \ln(L)} \\
&  \le C \ln(L) L^2 \int_{\R^2}  \left \lvert \left( \operatorname{d}_\L^{-1}((-2\varepsilon, -\varepsilon) ) \right)_x \right \rvert/ \varepsilon  \mathrm dx+ C(\Lambda ,p) L^2 \sqrt{ \ln(L)}  \\
&=  C \ln(L) L^2 \lvert \operatorname{d}_\L^{-1}((-2\varepsilon, -\varepsilon) ) \rvert/\varepsilon + C(\Lambda ,p) L^2 \sqrt{ \ln(L)}  \le C( \L, p) L^2 \ln(L)\,.
\end{align}
In the first step, we did a change of variables, in the third step we used \autoref{Def G}, in the last but one step Fubini and in in the final step we applied \eqref{tub L}.
\end{proof}

\section{The error term can be large and not smaller than \texorpdfstring{$o(L^2\ln(L))$}{\textit{o(L\texttwosuperior}ln(\textit{L}))}} \label{Section: lower order}

Without loss of generality we assume throughout this section that $\nu=1$ and $B=1$ because the precise values are not relevant now. The non-asymptotic bound in the following lemma is simple and useful in the proof of the main theorem in this section.
\begin{lemma} \label{lem: 1d bound}
Let $\Omega \subset \mathbb R$ be a finite union of intervals of finite lengths $\ell_1, \dots, \ell_n$ with disjoint closures. Let $m \in \mathbb N$ with $m \ge 2$, $\mu >0$, and $\D = \mathrm{d}^2/\mathrm{d}^2 x$ the one-dimensional Laplacian. Then we have the estimate
\begin{align}
\lVert (\mathds{1}_\Omega \mathds{1}(- \D \le \mu)\mathds{1}_\Omega)^m - \mathds{1}_\Omega \mathds{1}(- \D \le \mu)\mathds{1}_\Omega \rVert_1 \le \frac{m-1}{\pi^2} \sum_{j=1} ^n \ln( 1+ \sqrt\mu  \ell_j) + Cm n\,,
\end{align}
where $C$ is an entirely independent constant.
\end{lemma}

For $m=2$, this estimate is sharp in the sense that the prefactor $1/\pi^2$ equals the coefficient of the leading asymptotic behavior of $\tr (\mathds{1}_{L\Omega} \mathds{1}(- \D \le \mu)\mathds{1}_{L\Omega})^2$ for large $L$.

\begin{proof} By scaling we can assume $\mu=1$ since $\mathds{1}_{\Omega} \mathds{1}(- \D \le \mu)$ is unitarily equivalent to $\mathds{1}_{\sqrt{\mu}\Omega} \mathds{1}(- \D \le 1)$. In other words, we may set $\mu=1$ and eventually replace the lengths $\ell_j$ by $\sqrt{\mu}\ell_j$.

Then, we use the geometric series $a^m-a=a(a-1) (a^{m-2} + \cdots + a+1)$ with $a \coloneqq \mathds{1}_\Omega \mathds{1}(- \D \le 1)\mathds{1}_\Omega$. As $a$ has on operator norm of at most $1$, we can estimate
\begin{align}
\lVert (\mathds{1}_\Omega \mathds{1}(- \D \le 1)\mathds{1}_\Omega)^m &- \mathds{1}_\Omega \mathds{1}(- \D \le 1)\mathds{1}_\Omega \rVert_1 \nonumber
\\&\le  (m-1) \lVert a(a-1) \rVert_1 \nonumber\\
&= (m-1) \lVert \mathds{1}_\Omega \mathds{1}(- \D \le 1) \mathds{1}_{\Omega^\complement} \mathds{1}(- \D \le 1) \mathds{1}_\Omega \rVert_1 \nonumber\\
&= (m-1) \lVert \mathds{1}_\Omega \mathds{1}(- \D \le 1) \mathds{1}_{\Omega^\complement} \rVert_2 ^2 \nonumber\\
&= (m-1) \int_\Omega \mathrm{d} x \int_{\Omega^\complement} \mathrm{d}y \, k(x-y)^2 \label{A.6}\,,
\end{align}
with the function $k=k_1$ defined in \eqref{k mu}.

For a fixed $x \in \Omega$, we now enlarge the domain of integration in $y$ by allowing $y \in \Omega$, as long as $x$ and $y$ are in different intervals in $\Omega$. In a formula, with $\pi_0(\Omega)$ denoting the connected components (subintervals) of $\Omega$, the new domain of integration in \eqref{A.6} is
\be \bigcup _{ I \in \pi_0(\Omega)} \big\{(x,y) \colon x \in I, y \not \in I \big\}\,. \ee
As the integrand only depends on $x-y$, we may translate $I$ to be of the form $(0,\ell_j)$. Hence, with $n\coloneqq \#\pi_0(\O)$ the number of connected components of $\O$, we have
\begin{align}
\lVert (\mathds{1}_\Omega \mathds{1}&(- \D \le 1) \mathds{1}_\Omega)^m -  \mathds{1}_\Omega \mathds{1}(- \D \le 1)\mathds{1}_\Omega \rVert_1\\
&\le  (m-1) \sum_{j=1}^n \int_0^{\ell_j} \mathrm{d}x \int_{(0,\ell_j)^\complement} \mathrm{d}y \,k(x-y)^2 \\
&= (m-1) \sum_{j=1}^n \lVert \mathds{1}_{(0,\ell_j)} \mathds{1}(- \D \le 1) \mathds{1}_{(0,\ell_j)^\complement} \rVert_2^2 \\
&= (m-1) \sum_{j=1}^n \operatorname{tr} \Big[\mathds{1}_{(0,\ell_j)} \mathds{1}(- \D \le 1) \mathds{1}_{(0,\ell_j)} -\big(\mathds{1}_{(0,\ell_j)} \mathds{1}(- \D \le 1) \mathds{1}_{(0,\ell_j)}\big)^2\Big]  \\
&\le (m-1) \sum_{j=1}^n \frac1 {\pi^2} \ln(1+ \ell_j) + C mn \,.
\end{align}
The last step relies on an improved result by Landau and Widom with $L=1$, see \autoref{1D final asym}. 
\end{proof}

In \autoref{thm:polynomial}, we obtained for a general Lipschitz region $\L\subset\R^3$ an error term $o(L^2 \ln(L))$ and not of the order $L^2$. Specifically, using $\mathsf{I}(2) = -1/(4\pi^2)$, we have the asymptotic expansion
\begin{align}
\tr \big(\mathrm{D}_\mu(L\L)-\mathrm{D}_\mu(L\L)^2\big) &=  \frac{L^2 \ln(L)}{4 \pi^3}  \int_{\partial \L} \mathrm{d} \mathcal H^2( v)\, \lvert n(v) \cdot e_3 \rvert  +o(L^2\ln(L)) \,.
\end{align}
This allows us to define the error term $\varepsilon(L,\L)$ by the identity
\begin{align}
\tr \big(\mathrm{D}_\mu(L\L)-\mathrm{D}_\mu(L\L)^2\big) &=  \frac{ L^2 \ln(L)}{4 \pi^3}  \int_{\partial \L} \mathrm{d}{ \mathcal H^2}( v)\, \lvert n(v) \cdot e_3 \rvert  -L^2\ln(L) \,\varepsilon(L, \L)\,. \label{def error term}
\end{align}
In this notation, \autoref{thm:polynomial} states that $\lim_{L \to \infty}\varepsilon(L,\L)=0$ for a piecewise Lipschitz region $\L$ and we have $\sup_{L \ge 2} |\varepsilon(L, \L)|\ln(L) < \infty$, if $\L$ is a piecewise $\mathsf{C}^{1,\alpha}$ region. The main result of this section, which is the next theorem, shows that the estimate for Lipschitz regions is sharp and the error term can be large and just $o(L^2\ln(L))$. The negative sign in front of the error term does not necessarily mean that it has a definite sign although in our example it will be.  Although our result only deals with the error term for the simplest, non-trivial  polynomial, namely $t\mapsto t(1-t)$, we believe that also for the entropy the error term can be as large and only $o(L^2\ln(L))$ for a Lipschitz region.  

\begin{thm}\label{thm error term}
Let $\varphi \colon \R^+ \to \R^+$ be a bounded function with $\lim_{L \to \infty} \varphi(L)=0$. Then there is a piecewise Lipschitz region $\L$ and an $L_0$, such that for any $L \ge L_0$, the error term defined in \eqref{def error term} satisfies
\be \varepsilon(L, \L) \ge \varphi(L)\,. \ee
\end{thm}

\begin{remark}
Let $\mathcal A$ be the subset of the space of all polynomials vanishing at $0$ and $1$ such that the error term in \autoref{thm:polynomial} for $f \in \mathcal A$ is of order $O(L^2)$ for any Lipschitz domain $\L$. This is clearly a linear subspace and the theorem tells us that $t \mapsto t(1-t) \not \in \mathcal A$. Thus, the subspace has at least codimension one, which means that it satisfies (at least ) one linear constraint.  We conjecture that this constraint might be $f \in \mathcal A \implies \mathsf{I}(f)=0$. That is, the error term can only achieve the order $O(L^2)$, if the leading term of order $L^2 \ln(L)$ vanishes.
\end{remark}

\begin{proof}[Proof of \autoref{thm error term}]
We begin with a non-negative, summable sequence $(a_i)_{i \in \N}$ with $\sum_{i \in \N} a_i=1$, which we will choose later. Let $g_0 : [0,1] \to \R^+$ be the zigzag function defined by $g_0(0)=1$ and for $t>0$,
\begin{align}
g_0'(t)= \begin{cases} +1 & \text{if } \exists j \in \N \colon 0< t  -\sum_{i<j} a_i \le \frac 1 2 a_j \\
-1 & \text{if } \exists j \in \N \colon \frac 1 2 a_j < t  -\sum_{i<j} a_i \le  a_j
\end{cases}\,.
\end{align}
If $j=1$ then we use the convention that $\sum_{j<1}a_j \coloneqq 0$. Clearly, $g_0$ is Lipschitz continuous with Lipschitz bound $1$. We expand $g_0$ to $[-1,2]$ by setting $g_0(t)=t+1$ for $t < 0$ and $g_0(t)=2-t$ for $t>1$. This extension is still Lipschitz continuous and satisfies $g_0(-1)=g_0(2)=0$. Now, we can define the region $\L$,
\be \L \coloneqq \big\{ (x_1,x_2,x_3) \in \R^3  \colon x_1 \in (0,1), x_3 \in (-1,2), -g_0(x_3) < x_2 < g_0(x_3)\big\} \,.\ee
This clearly defines a piecewise Lipschitz region. We will now sketch why this is even a strong Lipschitz domain (see \cite[Pages 66--67]{Adams} for the definition.)

\begin{figure}[h] \label{ex domain}
\caption{Example of a $x_2$-$x_3$-plot of the domain $\L$ for any $x_1\in (0,1)$ and some sequence $(a_i)_{i \in \N}$. The upper half is the graph of $g_0$. In green one can see two sets $\L_{x^\perp}$. In the middle one can see the ball of all points, with respect to which $\L$ is star-shaped.}
\includegraphics[scale=0.6]{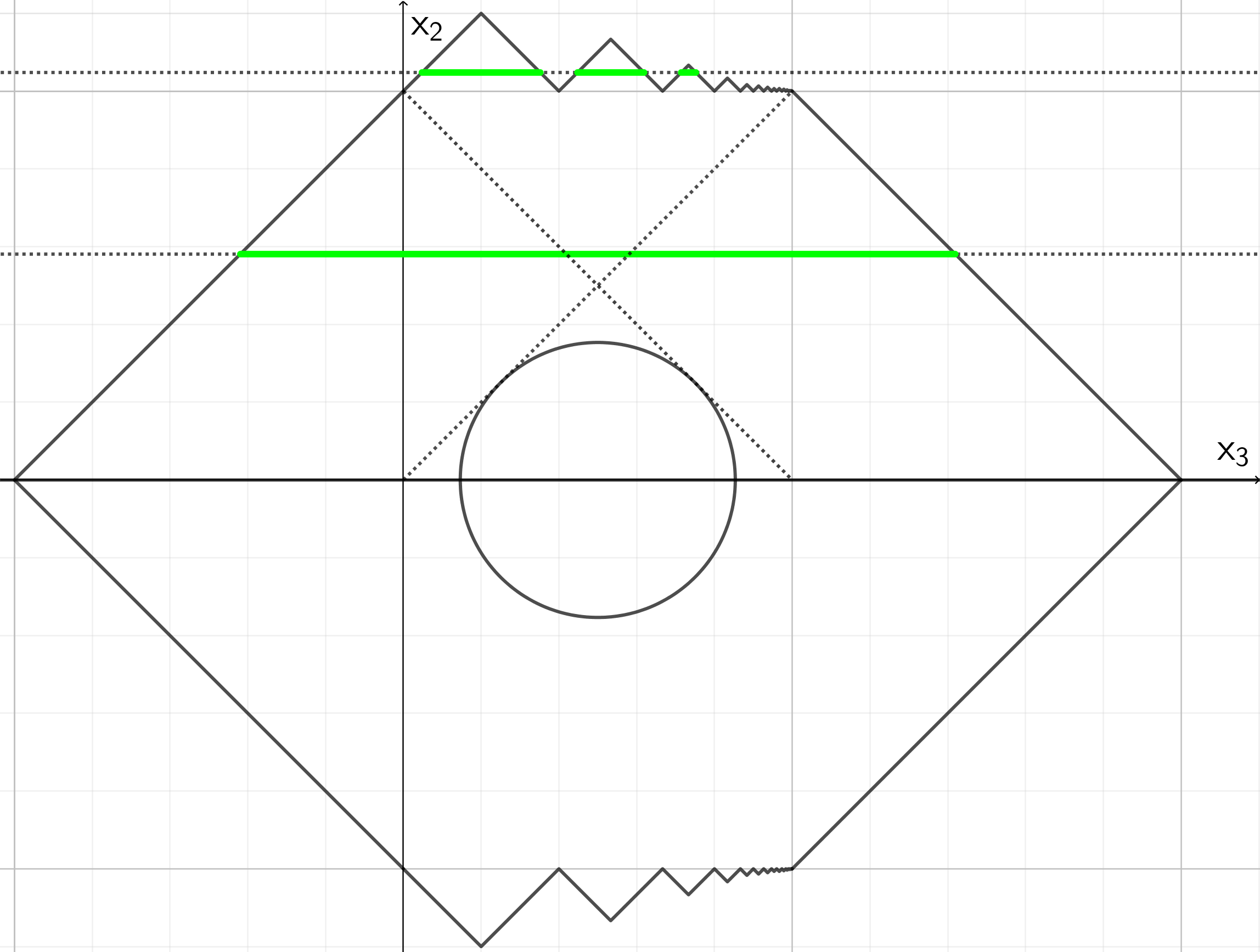}
\end{figure}

For any $x_0 \in B_{1/ (2\sqrt 2)} (1/2,0,1/2)$, the region $\L$ is star shaped with respect to $x_0$. For the definition of a strong Lipschitz domain, we need to choose an open cover of $\p \L$ and a projection with a certain direction on every set of the cover. For any orthogonal (rank 2) projection $\pi \colon \R^3 \to \R^2$, on the two connected components of the set $\pi^{-1} ( \pi (B_{1/  (4\sqrt 2 )} (1/2,0,1/2))) \cap \p \L$, one can define the chart as the inverse of $\pi$, which has a Lipschitz constant less than $10$. This leads to an open cover of $\p \L$ and one can then choose a finite subcover.

The boundary $\partial \L$ can be covered by the sets $\partial_1 \L  \coloneqq \big\{x \in \partial \L \colon x_1 \in \{0,1\}\big\}$ and $\partial_2 \L \coloneqq\big\{ x \in \partial \L\colon x_1 \in [0,1]\big\}$. These two boundary sets have a non-empty intersection, but $\partial_1 \L\cap\partial_2 \L$ is a ``one-dimensional" set with two-dimensional Hausdorff measure zero, that is, $\mathcal H^2(\partial_1 \L\cap\partial_2 \L)=0$.

For almost every $x \in \partial_1 \L$, the outward normal vector $n(x)$ is given by $\pm e_1$, while for almost every $x \in \partial_2 \L$, the outward normal vector is given by $\frac 1 {\sqrt 2} (\pm e_2 \pm e_3)$; the vectors ${e_1,e_2,e_3}$ are the usual unit vectors in the positive $x_1,x_2,x_3$ directions. Hence, we observe
\begin{align}
{\mathcal H^2}(  \partial_1 \L ) &= 4 \int_{-1}^2 g_0(t) \,\mathrm{d}t \le 9\,, \\
{\mathcal H^2}(  \partial_2 \L )&= 2 \int_{-1}^2 \sqrt {1+ (g_0')^2(t)} \,\mathrm{d}t = 6 \sqrt 2\,,\\
\int_{\partial \L} \lvert n(x) \cdot e_3 \rvert \,\mathrm{d}\mathcal H^2(x)&= \frac 1 {\sqrt 2} {\mathcal H^2}(  \partial_2 \L ) = 6\,. \label{ex cross volume}
\end{align}

It is important that ${\mathcal H^2}(  \partial \L ) = {\mathcal H^2}( \p_1\L) + {\mathcal H^2}( \p_2\L)$ is bounded independently of the sequence $(a_i)_{i \in \N}$ and that the surface integral in \eqref{ex cross volume} is completely independent of the sequence. 

The leading asymptotic term for the trace on the left-hand side of \eqref{def error term} is provided by \autoref{thm:polynomial}. Here, $\mathsf{I}(2) = -1/(4\pi^2)$ and hence
 \begin{align}
\tr \big(\mathrm{D}_\mu(L\L)-\mathrm{D}_\mu(L\L)^2\big) &=  \frac{L^2 \ln(L)}{4 \pi^3}  \int_{\partial \L} \mathrm{d}{ \mathcal H^2}( v)\, \lvert n(v) \cdot e_3 \rvert  +o(L^2\ln(L)) 
\\
&= \frac{L^2 \ln(L)}{\pi^3} \,\frac 3 2 + o(L^2\ln(L))\,,
 \end{align}
where we used \eqref{ex cross volume}.

We need an upper bound for the constant $\mathcal K( \L)$ defined in \autoref{tubular nbh small}, which is independent of the function $g$. We observe that  $\partial_1 \L$ is the image of the Lipschitz functions $f_j \colon [0,1]^2 \to \partial_1\L; (x_1,x_2) \mapsto (j,3x_1-1,g(3x_1-1)x_2)$ for $j=0,1$ and $\partial_2 \L$ is in the image of the Lipschitz functions $\tilde{f}_\pm \colon [0,1]^2 \to \partial_2 \L; (x_1,x_2) \mapsto (x_1,3x_2-1,\pm g(3x_2-1))$. Thus, the set $\{f_0,f_1\tilde {f}_+, \tilde{f}_-\}$ defines a piecewise Lipschitz atlas of $\p \L$. Hence, as $C_{\text{lip}}(f_j)=3 \sqrt 2$ and $C_{\text{lip}}(\tilde{f}_\pm)=3$, we observe 
\be \mathcal K(\L) \le (16 \sqrt 2)^2 \times 2 \times (1+(3\sqrt 2)^2 +1+3^2)< \infty \,.\ee
Hence, by \autoref{lem:reduction}, we have 
\begin{align}
 \tr \mathrm{D}_\mu(L\L)^m =\frac {L^2} {2 \pi} \int_{\R^2} \mathrm{d}x^\perp \,\operatorname{tr} \left( \mathds{1}_{L \L_{x^\perp} }\mathds{1}[(-\mathrm{i}\nabla^\parallel)^2 \le \mu]  \,\mathds{1}_{L \L_{x^\perp} } \right)^m  + O( L^2 ) \, \label{6.18}
\end{align}
with $x^\perp=(x_1,x_2)$. To get to the polynomial $f(t)=t(1-t)$ we have to subtract this term with $m=2$ from the term with $m=1$. 
Now, we intend to use \autoref{lem: 1d bound}. To do so, we need to describe the lengths of the (sub)intervals of $\L_{x^\perp}$ depending on $x^\perp=(x_1,x_2)$. 

We can ignore the case $x_1 \in \{0,1\}$, as this is a null set with respect to the Lebesgue measure on $\R^2$. If $x_1 \in (0,1)$ and $\lvert x_2 \rvert \le 1$ then the set $\L_{x^\perp}$ is a single interval of length $\ell_1(x_2)=3-2 \lvert x_2 \rvert$. The interesting case is $x_1 \in (0,1)$ and $1 < \lvert x_2 \rvert <2$. Here, for any $i \in \mathbb N$ with $a_i>2(\lvert x_2 \rvert-1)$, there is an interval of size $\ell_i(x_2)= a_i - 2 (\lvert x_2 \rvert -1)$, as illustrated in \autoref{ex domain}. For any $\lvert x_2 \rvert>1$, this will only lead to finitely many intervals, as the sequence $(a_i)_{i \in \N}$ is a null sequence. Now, we apply \eqref{6.18} for $m=1$ and $m=2$, and then \autoref{lem: 1d bound} and see that
\begin{align}
\frac{ 2 \pi} {L^2} &\Big| \tr \big( \mathrm{D}_\mu(L\L)-\mathrm{D}_\mu(L\L)^2\big) \Big|\\
 &\le {C+} \int_0^1 \mathrm{d}x_1 \int_\R \mathrm{d}x_2 \, \lVert (\mathds{1}_{L\L_{x^\perp}} \mathds{1}(-\Delta \le \mu)\mathds{1}_{L\L_{x^\perp}})^2 - \mathds{1}_{L\L_{x^\perp}}\mathds{1}(- \Delta \le \mu) \mathds{1}_{L\L_{x^\perp}} \rVert_1 \\
&\le  \frac 2 {\pi^2} \int_0^1 \mathrm{d} x_2\, \big(\ln(1+L\sqrt{\mu}(3-2x_2)) + C\big)  \\
 & + \frac 2 {\pi^2}\int_1^2 \mathrm{d}x_2\, \sum_{i \in \N \colon a_i > 2 (x_2-1)} \big(\ln(1+  L \sqrt{\mu}(a_i-2(x_2-1))+C\big) \\
&=\frac 2 {\pi^2} \int_0^1  \mathrm{d} x_2 \,\Big(\ln(1+L\sqrt{\mu}(3-2x_2)) +C\Big) + \frac 2 {\pi^2}\sum_{i \in \N} \int_{0}^{\frac 1 2 a_i} \mathrm{d}t\, \Big(\ln(1+  L \sqrt{\mu}2t) +C\Big) \,.
\end{align}
In the second step we also used that the set $\L_{x^\perp}$ is independent of $x_1 \in (0,1)$. The third step uses Fubini 
to exchange the sum and the integral and then transforms the integration variable to $t \coloneqq \frac 1 2 a_i+1- x_2$. The lower bound $0$ in the last integral stems from the condition $a_i>2(x_2-1)$, respectively from $t>0$.

We intend to show, that this upper bound is significantly smaller than the known asymptotics. The difference between the asymptotics and this upper bound can then be used as a lower bound for the error term. This is why it is very important that the coefficient in front of the upper bound is equal to the asymptotic coefficient and is thus the reason why we can only do this here for the polynomial $f(t)=t(1-t)$. 
 
We now allow our constants to depend on $\mu$ (for general $\nu\ge1$, they depend on all values of $\mu(\ell)$) and use the trivial inequality $\ln(1+ab) \le \ln(1+a) + \ln(1+b)$ for $a,b\ge0$ to arrive at
\begin{align}
\frac{\pi^3} {L^2} &\left|\tr \mathrm{D}_\mu(L\L)-\mathrm{D}_\mu(L\L)^2 \right| \\
&\le  \int_0^1  (\ln(1+L) + C )\,{\mathrm{d} x_2 \, }  + \sum_{i \in \N} \Big[\int_0^{\frac 1 2 a_i} \ln(1+  L t)\, \mathrm{d}t + C a_i \Big] {+ C}\\
&= \ln(1+L) + \frac 1 L \sum_{i \in \N} {\Big(}\Big(1+ \frac 1 2 a_iL\Big) \Big( \ln\Big(1+\frac 12 a_i L\Big) -1\Big) -(-1) \Big) + C   \\
&=  \ln(1+L) + \sum_{i \in \N} \Big[\frac 1 2 a_i \ln\Big(1+\frac 12 a_i L\Big)+\frac {\ln(1+\frac 12 a_i L)} L -\frac 1 2 a_i \Big] +C   \\
&\le  \ln(L) + \sum_{i \in \N} \frac 1 2 a_i \ln\Big(1+\frac 12 a_i L\Big)  +C \,,
\end{align}
or equivalently,
\begin{align} \label{non asympt bound}
\big(0\le\big)\,\tr \big(\mathrm{D}_\mu(L\L)-\mathrm{D}_\mu(L\L)^2\big) \le \frac{L^2\ln(L)}{\pi^3} \left[1+ \frac{1}{\ln(L)}\sum_{i \in \N} \frac 1 2 a_i \ln\Big(1+\frac 12 a_i L\Big)  + \frac{C}{\ln(L)}\right]\,.
\end{align}
In the first step, we used $1 \le 3- x_2 \le 3$, and in the fourth step, we used $\ln(1+L)\le \ln(L)+1$ for $L \ge 1$ and $\ln(1+\frac 12 a_i L) \le \frac 12 a_i L$.

Now we rewrite \eqref{def error term} and use \eqref{ex cross volume} and \eqref{non asympt bound} to obtain
\begin{align}
2 \pi^3 \varepsilon(L, \L) &= \frac{1}{2} \int_{\p\L} \mathrm{d}\mathcal H^2(v) \, |n(v)\cdot e_3| - \frac{2\pi^3}{L^2\ln(L)} \tr\big(\mathrm{D}_\mu(L\L) - \mathrm{D}_\mu(L\L)^2\big)
\\
&\ge 3   - {\left( 2+ \frac{1}{\ln(L)} \sum_{i \in \N}  a_i \ln\Big(1+\frac 12 a_i L\Big)     + \frac{2C}{\ln(L)}    \right) } \\
& =1  - \frac 1 {\ln(L)} \sum_{i \in \N} a_i \ln\Big(1+\frac 12 a_i L\Big)      - \frac C { \ln(L)}  \\
& = -\frac 1 {\ln(L)} \sum_{i \in \N}  a_i \ln\left( \frac 1 L+\frac 12 a_i \right)      - \frac C { \ln(L)}  \\  
&\ge -\frac 1 {\ln(L)} \sum_{i \in \N, a_i < \frac 1 L }  a_i \ln\left( \frac {3} {2L} \right)      - \frac C { \ln(L)}  \\ 
&\ge \sum_{i \in \N, a_i <  \frac 1 L }  a_i      - \frac C { \ln(L)} \eqqcolon \varepsilon_0(L)\,. \label{def eps0(L)}
 \end{align}
 The fourth step uses $\sum_i a_i=1$ and the fifth step relies on $L \ge 2, a_i \le 1$ to get $\ln( \frac 1 L+\frac 12 a_i ) \le 0 $. In the last step, $C$ changed. Now, we just need to find a good sequence $(a_i)_{i \in \N}$. To show our claim, it suffices to find a sequence $(a_i)_{i \in \N}$ such that
 \begin{align}
 \lim_{L \to \infty} \varphi(L) / \varepsilon_0(L) =0\,,
 \end{align}
since then the quotient $\varphi(L)/\varepsilon(L,\L) \le 2 \pi^3\varphi(L)/\varepsilon_0(L)\to0$ is less than $1$ for large $L$, that is, $\varepsilon(L,\L)\ge \varphi(L)\ge0$ for $L\ge L_0$, where $L_0$ is chosen below.

The construction of the sequence $a_i$ relies on \autoref{technical function existence}, and we apply this Lemma with $f$ as $\varphi$. With the resulting function $\mathrm{Env}(\varphi)$ we define the sequence of real numbers $a_i \coloneqq \mathrm{Env}(\varphi)(i-1)-\mathrm{Env}(\varphi)(i)$ for $i \in \N$. 
 As $\lim_{L \to \infty} \mathrm{Env}(\varphi)(L)=0$, we have $\sum_{i \ge L} a_i=\mathrm{Env}(\varphi)(L)$ and in particular $\sum_{i \in \N}a_i=\mathrm{Env}(\varphi)(0)=1$. As $\mathrm{Env}(\varphi)$ is non-increasing and convex, the $a_i$ are non-negative and non-increasing.  
 As the sequence defined this way is non-increasing and $\sum_{i\in\N} a_i =1$, we have $a_i \le \frac 1 i  \sum_{j \le i}a_j \le \frac 1 i$. Hence, we know that $i \ge L$ implies $a_i \le \frac 1 L$. Thus, we have the estimate
 \be \sum_{i \in \N, a_i  \le \frac 1 L }  a_i\ge \sum_{i \ge L} a_i= \mathrm{Env}(\varphi)(L) \,.\ee 
Furthermore, as $\mathrm{Env}(\varphi)(L) \ge C/\sqrt{\ln(2+L)}$, for $L$ large enough, we have
 \be \mathrm{Env}(\varphi)(L) - \frac C{\ln(L)} \ge \frac 1 2 \mathrm{Env}(\varphi)(L) \,. \ee
Hence, we conclude that
  \begin{align}
0\le  \lim_{L \to \infty} \frac{\varphi(L)}{\varepsilon_0(L)} =  \lim_{L \to \infty} \frac{\varphi(L)}{\sum_{i \in \N, a_i \le \frac 1 L } a_i - \frac C { \ln(L)}} \le 2\lim_{L\to \infty}\frac{\varphi(L)}{\mathrm{Env}(\varphi)(L)}=0\,. 
 \end{align}
One choice of $L_0$ could be that $4\varphi(L) \le \mathrm{Env}(\varphi)(L)$ for $L>L_0$ is satisfied. Thus, by this and \eqref{def eps0(L)}, there is an $L_0 >0$ such that for any $L>L_0$, we have
 \be \varepsilon(L,\L) \ge \varepsilon_0(L) /(2\pi^3) \ge \varphi(L)\,. \ee 
This finishes the proof.
\end{proof}

\begin{appendix}

\section{Some geometric results}\label{Appendix:B}

Here, we assemble a few geometric statements that we used.
\begin{lemma} \label{lx estimate}
Let $\L \subset \R^{d+1}$ be a piecewise $\mathsf{C}^{1,\alpha}$ region for some $0 < \alpha <1$. Let  $(\Psi_{\mathrm{pC},i})_{i \in I}$ be a piecewise $\mathsf{C}^{1,\alpha}$ atlas of $\p\L$  and $\Gamma$ as defined in \autoref{def ps}. Then there is a constant $C$ depending only on $\L$ such that for all unequal $v_1$ and $v_2$ in  $\partial \L$, we have
 \be \|v_1 -v_2 \| \ge C \min \left \{ \left \lvert n(v_1) \cdot \frac {v_1-v_2}{\|v_1-v_2\|} \right \rvert^{\frac 1 {\alpha}} , \operatorname{dist}(v_1, \Gamma)\right\} \,.\ee
\end{lemma}

\begin{remark} The normal vector $n(v_1)$ is well-defined if $v_1 \not \in \Gamma$. In the case $v_1 \in \Gamma$, the minimum on the right-hand side is meant to be $0=\operatorname{dist}(v_1, \Gamma)$, which turns it into a trivial statement.
\end{remark}

\begin{proof}  We begin with the case $v_1 \in \Gamma$ or $v_2 \in \Gamma$; $v_1\in\Gamma$ is explained in the above remark. If $v_2\in\Gamma$, then we trivially have $\lVert v_1 - v_2 \rVert \ge \operatorname{dist}(v_1, \Gamma)$ and thus the claim holds for any $C \le 1$.

Let us now consider the case that there is an $i\in I$ such that both $v_1$ and $v_2$ are in $\Psi_{\text{pC},i}((0,1)^d)$. As $\Psi\coloneqq \Psi_{\text{pC},i}$ is injective, there are unique $x_k\in (0,1)^d$ such that $\Psi(x_k)=v_k\in\R^{d+1}$ for $k=1,2$. We observe that $n(v_1) \cdot D\Psi(x_1)=0 \in \R^d$, as the image of the matrix $D\Psi(x_1)$ is the tangent space to $\partial \L$ at $v_1$ and hence is orthogonal to the outward normal vector $n(v_1)$. Thus, using \eqref{C1a cond}, we see 
\begin{align}
\lvert n(v_1) \cdot (v_1-v_2) \rvert &= \big\lvert n(v_1) \cdot \big(\Psi(x_1)-\Psi(x_2)\big)  \big\rvert \\
&\le \lVert x_1 - x_2 \rVert \sup_{t \in [0,1]} \big\lvert n(v_1) \cdot D\Psi(tx_1+ (1-t)x_2) \big\rvert \\
&=  \lVert x_1 - x_2 \rVert \sup_{t \in [0,1]} \big\lvert n(v_1) \cdot (D\Psi(tx_1+ (1-t)x_2) - D\Psi(x_1) \big\rvert \\
&\le \lVert x_1 -x_2 \rVert \, C \lVert x_1-x_2 \rVert^\alpha = C \| x_1-x_2 \|^{1+\alpha}\,.
\end{align}
As $\Psi$ is bi-Lipschitz, we know that $\|v_1- v_2 \| = \|\Psi(x_1)-\Psi(x_2) \| \ge C \|x_1 -x_2 \|$. Using this and dividing both sides by $\|v_1-v_2\|$, we arrive at
\be \left \lvert n(v_1) \cdot \frac {v_1-v_2}{\|v_1-v_2\|} \right \rvert \le C\|v_1-v_2 \|^\alpha\,. \ee

We are now in the remaining case that $v_1$ and $v_2$ lie in different $\Psi_{\text{pC},i}((0,1)^d)$'s since $\p\Lambda = \Gamma \cup \bigcup_{i\in I} \Psi_{\text{pC},i}((0,1)^d)$.

Let $(\Psi_{\text{gL},j})_{j\in J}$ be a global Lipschitz atlas of $\p \L$. As  $\p\L = \bigcup_{j\in J} \Psi_{\text{gL},j}((0,1)^d)$ is a cover by (relatively) open sets and $\p\L$ is a compact metric space, by Lebesgue's number lemma, there is an $\varepsilon>0$ such that for all $v \in \p \L$ there is an $j\in J$ with $B_{\varepsilon}(v) \cap \p \L \subset \Psi_{\text{gL},j}((0,1)^d)$, where $B_{\varepsilon}(v)\subset \R^{d+1}$ is the open ball of radius $\varepsilon$ at $v$.

If $\lVert v_1- v_2 \rVert \ge\varepsilon$, we can choose $C= \varepsilon$ to get the statement, as the first expression inside the minimum is at most $1$. 

Hence, we are left with the case $\| v_1- v_2 \| < \varepsilon$. Now, we get an $j \in J$ such that $v_1,v_2 \in \Psi_{\text{gL},j}((0,1)^d)$.  Again, we define $y_k$ by $\Psi_{\text{gL},j}(y_k)=v_k$ for $k=1,2$. The image $\gamma$ of the linear path from $y_1$ to $y_2$ is at most $C \|y_1-y_2\|\le C \|v_1 -v_2 \|$ long. As $v_1$ and $v_2$ are in the images of two different $\Psi_{\text{pC},i}$'s, the path $\gamma$ has to intersect some edge $\Psi_{\text{pC},i}(\p (0,1)^d)$ which implies $\gamma \cap \Gamma \neq \emptyset$. Hence, we have
\be \operatorname{dist}(v_1, \Gamma) \le {\mathcal H^1}(  \gamma ) \le C \|y_1- y_2 \| \le C \|v_1 -v_2 \|\,. \ee
The last inequality follows since $\Psi_{\text{gL},j}$ is bi-Lipschitz. This finishes the proof.
\end{proof}  

\begin{lemma} \label{tubular nbh small}
For $d\ge 1$, let $f \colon [0,1]^d \to \R^{d+1}$ be a Lipschitz continuous function with Lipschitz constant $C_{\mathrm{lip}}(f)$ and let $\L \subset \R^{d+1}$ be a piecewise Lipschitz region. Then for any $r>0$, the $(d+1)$-dimensional Lebesgue volume of the $r$-neighborhood (see \eqref{def: ball}) of the set $f([0,1]^d)$ in $\R^{d+1}$ satisfies
\be \lvert B_r (f([0,1]^d)) \rvert \le {(16 \sqrt d)^d (C_{\mathrm{lip}}(f)^d+1)   }(r+r^{d+1} )\,,\ee
and the set $\p\L$ satisfies the bounds
\begin{align}
 \lvert B_r (\partial \L) \rvert &\le {\mathcal K(\L)}(r+r^{d+1} )\,,\label{tub L} \\
 \mathcal H^d(\p \L) &\le \mathcal K(\L)\,,
\end{align}
where $\mathcal K(\L)$ is described  as follows: Let $\mathcal A$ be the set of all piecewise Lipschitz atlases of $\partial \L$, as defined in \autoref{def ps}. Then, we define
\be \mathcal K (\L) \coloneqq \inf_{(\Psi_{\mathrm{pL},i})_{i \in I}  \in \mathcal A} \sum_{i \in I}  (16 \sqrt d)^d (C_{\mathrm{lip}}(\Psi_{\mathrm{pL},i})^d+1) \,.\ee 
\end{lemma}

\begin{proof}
 We consider the set
\begin{align}
A_r \coloneqq  \left(\frac {r}{C_{\text{lip}}(f)\sqrt d} \Z \right)^d \cap [0,1]^d\,.
\end{align}
The maximum distance a point in $[0,1]^d$ can have from $A_r$ is less than $\frac {r}{C_{\text{lip}}(f)}$.  For the cardinality $\#A_r$ of $A_r$, we observe
\be \#A_r \le \left( 1 +\frac{C_{\text{lip}}(f)\sqrt d} r \right)^d\le 2^{d-1} \left( 1+ \frac{C_{\text{lip}}(f)^d \sqrt d^d}{ r^d} \right) \le 2^{d-1} \sqrt d ^d  (1+C_{\text{lip}}(f)^d )(1+r^{-d})\,.\ee 
For any $x \in [0,1]^d$, there is a $z \in A_r$ such that $\| x-z\| \le r /C_{\text{lip}}(f)$ and thus $\|f(x)- f(z) \|  \le r$. This implies $B_r(f(A_r)) \supset f([0,1]^d)$, which leads to $B_{2r}(f(A_r)) \supset B_r(f([0,1]^d))$. Hence, we get
\begin{align*}  \lvert B_r (f([0,1]^d)) \rvert &\le \lvert B_{2r}(f(A_r)) \rvert \le {\lvert B_1(0) \rvert} \, \# A_r (2r)^{d+1} \le  4^{d+1} \,\#A_r r^{d+1} 
\\
&\le (16 \sqrt d)^d(1+C_{\text{lip}}(f)^d )    (r+r^{d+1})\,. \end{align*} 
This finishes the proof of the first statement. The second statement is trivially implied by the first one. Furthermore, as $B_r(f(A_r)) \supset f([0,1]^d)$ due to the definition of the Hausdorff measure (see e.g. \cite[Definition 2.1]{Evans}) we observe
\begin{equation*}
\mathcal H^d(f([0,1]^d) \le \lim_{r \to 0} \lvert B_1^{(d)}(0) \rvert \#A_r r^d  \le \lim_{r \to 0} (4 \sqrt d)^d (1+C_{\text{lip}}(f)^d )(1+r^{-d}) r^d =  (4 \sqrt d)^d (1+C_{\text{lip}}(f)^d ) \, .
\end{equation*}
The final statement is a corollary of this inequality.  We want to note that $\mathcal K(\L)<\infty$ for any piecewise Lipschitz region $\L$, as we require in this paper our atlases to be a finite collection of charts.
\end{proof}

\begin{lemma} \label{normal vector exists lem}
Let $\L \subset \R^3$ be a piecewise Lipschitz region with piecewise Lipschitz atlas $(\Psi_{\mathrm{pL},i})_{i \in I}$. Let $v_0 \in \p\L$ satisfy that there are $i_0 \in I$ and $x_0 \in (0,1)^2$ such that $\Psi_{\mathrm{pL},i_0}(x_0)=v_0$ and the Jacobi matrix $D \Psi_{\mathrm{pL},i_0}(x_0)$ exists. Then the signed distance function $\operatorname{d}_\L$ is differentiable at $v_0$, the outward unit normal vector $n(v_0)$ is well-defined, orthogonal to the image of $D \Psi_{\mathrm{pL},i_0}(x_0)$, and $D\operatorname{d}_\L(v_0)=n(v_0)$.
\end{lemma}

To prove this statement, we need the following result from intersection theory.
\begin{lemma} \label{intersection lem}
Let $R>0$, $f_1 \colon [-1,1] \to \overline{ B}_R^{(3)} (0), f_2 \colon \overline{ B}_R^{(2)} (0) \to \overline{ B}_R^{(3)} (0)$ be continuous functions such that $f_1( \pm 1 ) = \pm R e_3$ and $f_2$ restricted to the boundary is the equatorial embedding, that is, $f_2(x)=(x,0)$ for $\|x \|=R$. Then, the images of $f_1$ and $f_2$ intersect.
\end{lemma}
\begin{proof}
Without loss of generality, we assume $R=1$. Assume $f_1$ and $f_2$ were two such functions such that their images do not intersect. Let $\eta_1 \colon \R^1 \to \R^3, t \mapsto (0,0,t)$ and $\eta_2 \colon \R^2 \to \R^3, x \mapsto (x,0)$ be the natural orthogonal inclusions. The assumptions on $f_j$ can now be stated as $f_j(x_j)= \eta_j(x_j)$ for $x_j \in \R^j$ with $\|x_j \|=1$ for $j\in\{1,2\}$. We extend the maps $f_j$ to the $\R^j$ by setting
\begin{align}
 f_j(x_j) \coloneqq \begin{cases} f_j(x_j) &\quad \text{ if } \| x_j \| \le 1 \\ \eta_j(x_j) &\quad \text{ if } \| x_j \|  >1 \end{cases} \,,\quad x_j \in \R^j\,,\,j\in\{1,2\}\,.
 \end{align} 
Trivially, these extensions are still continuous and their images still do not intersect. As the images do not intersect and only get close to each other in the compact set $\overline{B}_1^{(3)}(0)$, they have a positive distance. We can now mollify $f_j$ by convolution with an appropriately chosen, compactly supported smooth function to get $\hat f_j$ such that the images of $\hat f_1$ and $\hat f_2$ still have positive distance and $\hat f_j (x_j)=\eta_j(x_j)$ for any $x_j \in \R^j$ with $\| x_j \| \ge 2$.

For $d=1,2,3$, consider the  sphere $\mathbb S^d= \R^d \cup \{\infty\}$. With the charts $\text{id} \colon \R^d \to \mathbb S^d, x \mapsto x$ and $\iota_d \colon \R^d \to \mathbb S^d, x \mapsto x / \lVert x \rVert^2, 0 \mapsto \infty$, it becomes a differentiable manifold. We now extend $\hat f_j$ to a function from $\mathbb S^j$ to $\mathbb S^3$ by setting $\hat f_j(\infty)\coloneqq\infty$ for $j\in\{1,2\}$. The point $\infty$ is now an intersection point of $\hat f_1$ and $\hat f_2$. We want to show that the extended functions are still smooth and that they intersect transversely at $\infty$, see for instance \cite[Page 113]{DifftopGP}. For $j\in\{1,2\}$ and $x_j \in \R^j$ with $\| x_j \| < \frac 1 2$, we observe
\be 
(\iota_3^{-1} \circ \hat f_j \circ \iota_j)(x_j)= (\iota_3^{-1} \circ \eta_j )(x_j \lVert x_j \rVert^{-2})=\eta_j(x_j)\,.
\ee
Thus, in the charts $\iota_j ,\iota_3$ the maps $\hat f_j$ are linear and orthogonal at $0$ (which corresponds to $\infty \in \mathbb S^j$). The maps $\hat f_1,\hat f_2$ are therefore smooth and intersect transversely at $\infty$. In conclusion, we have just constructed two smooth maps $\hat f_j \colon \mathbb S^j \to \mathbb S^3$, which intersect transversely and have a unique intersection point. Thus, their oriented intersection number is equal to the local intersection number at this intersection point, which is $+1$ or $-1$ (in fact, it is $+1$). However, both maps are contractible (homotopic to a constant map) and thus, as oriented intersection numbers are homotopy invariant (see \cite[Page 115]{DifftopGP}), they should have intersection number $0$. This is a contradiction. Hence, the assumption that $f_1$ and $f_2$ do not intersect was wrong. 

\end{proof}

\begin{proof}[Proof of \autoref{normal vector exists lem}]
Let $C_{\text{lip}}$ be a bi-Lipschitz constant of $\Psi_{\text{pL},i_0}$, as in \autoref{biLip def}. Then, for any $x \in \R^2$ with $\lVert x \rVert=1$, we observe $C_{\text{lip}}^{-1} \le \lVert D\Psi_{\text{pL},i_0}(x_0) x \rVert \le C_{\text{lip}}$. This means that the Jacobi matrix is invertible. Thus, using affine linear transformations on $\R^2$ and on $\R^3$, we can transform the function $\Psi_{\text{pL},i_0}$ into a function $\Psi$ such that $x_0,v_0$ are mapped to $0$\footnote{in $\R^2$ resp. $\R^3$} and the Jacobi matrix turns into the standard inclusion $J \colon \R^2 \to \R^3, x \mapsto (x,0)$. The function $\Psi$ is now defined on some closed parallelogram $P$ containing $0$ in its interior $P^{\text{int}}$. Let $C_{\text{lip}}\ge 2$ be a bi-Lipschitz constant for $\Psi$. Let $0<\varepsilon <1/2$. Then, there is an $r>0$ such that 
\begin{itemize}
\item For any $x \in B_{2C_{\text{lip}}r}^{(2)}(0)$, we have $\lVert \Psi(x)-(x,0) \rVert \le \varepsilon \lVert x \rVert$, as $D \Psi(0)=J$;
\item $B_{3r}^{(3)}(0)\cap \p \L \subset \Psi(P^{\text{int}})$, as $\Psi^{\text{int}}$ is relatively open in $\p \L$;
\item The set $B_{r}^{(3)}(0)\cap (\p \L)^\complement$ has exactly two connected components, as $\overline{\L}$ and $\L^\complement$ are topological manifolds with common boundary $\p\L$.
\end{itemize}
As $\Psi$ is bi-Lipschitz, we observe
\begin{align}
B_{2r}^{(3)}(0) \cap \p \L &\subset \Psi\left(B_{2 C_{\text{lip}}r}^{(2)}(0) \right) \subset \bigcup_{x \in  B_{2 C_{\text{lip}}r}^{(2)}(0)} 
 \overline{B}_{\varepsilon \lVert x \rVert}^{(3)} ((x,0))\\
  &\subset  \bigcup_{x \in \R^2 } \overline{ B}_{\varepsilon \lVert x \rVert}^{(3)} ((x,0))  =\big\{ v \in \R^3 \colon \lvert v \cdot e_3 \rvert \le \varepsilon \| v \| \big\} \, . \label{U0 inclusion}
\end{align} 
We define
\begin{align}
	U_0 &\coloneqq \big\{ v \in \R^3 \colon \lvert v \cdot e_3 \rvert \le \varepsilon \| v \| \big\}\,, \\
	U_\pm &\coloneqq \big\{ v \in \R^3 \colon \pm  v \cdot e_3 >  \varepsilon \| v \| \big\}\,.
\end{align}
The sets $U_\pm \cap B_r^{(3)}(0)$ are open, convex and do not intersect $\p\L$ due to \eqref{U0 inclusion}. We will now use \autoref{intersection lem} to show that, up to a binary choice, we may assume $U_- \cap B_r^{(3)}(0) \subset \L$ and $U_+ \cap B_r^{(3)}(0) \subset \L^\complement$. As $B_r^{(3)}(0) \cap (\p\L)^\complement$ has exactly two connected components,  it is sufficient to prove that any (continuous) path $p \colon [-1/2 , 1/2 ] \to B_r^{(3)}(0)$ with $p(\pm 1/2) \in U_{\pm}$ intersects $\p\L$. 

We first use the convexity to extend $p$ by an (affine) linear path at both ends to get a path $f_1 \colon [-1,1] \to \overline{B}_{ 5 r}^{(3)}(0)$ with $f_1(\pm 1)= \pm 5 r e_3$. Then we define $f_2  \colon \overline{B}_{ 5  r}^{(2)}(0) \to \R^3$ by 
\begin{align}
	f_2(x) \coloneqq
	\begin{cases}
		\Psi(x) \quad & \text{ if } \lVert x \rVert < 2 r  \\
		\frac{\lVert x \rVert -3r }{r } (x,0) + \frac{ 3 r - \lVert x \rVert}{r } \Psi(x) \quad & \text{ if } 2 r\le \lVert x \rVert < 3   r \\
		(x,0) & \text{ if } 3r  \le  \lVert x \rVert \le   5   r
	\end{cases} \,.\label{hat Psi eps def}
\end{align}
We see that $f_2 $ is Lipschitz continuous. The middle case is just a convex combination between the two other cases. 
Let $x \in \R^2$ with $\lVert x \rVert \le 3 r$. We observe
\begin{align}
	\lVert f_2 (x) - (x,0) \rVert \le \sup_{t \in [0,1] } \lVert t \Psi(x) + (1-t) (x,0)  -(x,0) \rVert \le \sup_{t \in [0,1]}  t \lVert \Psi(x)- (x,0) \rVert 
	<  \varepsilon \lVert x \rVert\,.	 \label{hat Psi eps - x}
\end{align}
This implies 
\[ |f_2(x)\cdot e_3| = |f_2(x)\cdot e_3 -(x,0)\cdot e_3| \le \|f_2(x)-(x,0)\|<\varepsilon \lVert x \rVert\,,
\]
that is, $f_2\left(\overline{B}_{3r}^{(2)}(0)\right) \subset U_0$ and thus, by the definition of $f_2(x)$ for $\|x\|\ge 3r$, $f_2 \left(\overline {B}_{5r}^{(2)}(0) \right) \subset U_0$. By the triangle inequality and with $\varepsilon < \frac 12 $ we obtain
\be 
	\frac 1 2 \lVert x \rVert < \lVert f_2(x) \rVert  <\frac 3 2 \lVert x \rVert \, .
\ee
These inequalities yield $ f_2^{-1} \left(B_{r}^{(3)} (0) \right) \subset  B_{2r}^{(2)}(0)$ and $f_2 \left( \overline{B}_{3r}^{(2)} (0)\right)  \subset B_{5r}^{(3)}(0)$. The latter inclusion together with the definition of $f_2$ outside $B_{3r}^{(2)}(0)$ implies $f_2\left( B_{5r}^{(2)}(0) \right) \subset B_{5r}^{(3)}(0)$. Thus, $f_1$ and $f_2$ satisfy the assumptions of \autoref{intersection lem} (with $R \coloneqq 5 r$) and consequently, they have an intersection point $s \in \R^3$. We have $s \in f_2\left(\overline{  B}_{5r}^{(2)}(0)\right) \subset U_0$ and $f_1^{-1}(U_0) \subset (-\frac 12, \frac 12)$, which means $s$ is in the image of the original path $p$. Thus, $s \in B_r^{(3)}(0)$, which implies that $s \in f_2\left(B_{2r}^{(2)}(0)\right) = \Psi \left(B_{2r}^{(2)}(0)\right)\subset \p \L$. Therefore, the path $p$ intersects $\p\L$, which was our claim.

As a result, we know that the sets $U_\pm \cap B_r^{(3)}(0)$ lie on opposite sides of $\p\L$. Without loss of generality, we assume $U_- \cap B_r^{(3)}(0) \subset \L$ and  $U_+ \cap B_r^{(3)}(0) \subset \L^\complement$. In terms of the signed distance function $\operatorname{d}_\L$, this means that $\pm \operatorname{d}_\L(v) >0$ for $v \in  U_\pm \cap B_r^{(3)}(0)$.

We are left to estimate $\operatorname {dist}(v, \p \L)$ for $v \in B_r^{(3)}(0)$. We start with the case $v \in U_0 \cap B_r^{(3)}(0)$. For that, we consider the map
\be
\Phi \colon B_r^{(2)}(0) \times [-2 \varepsilon , 2 \varepsilon] \to \R^3\,, \quad (y,t) \mapsto (\sqrt{1-t^2} y, t \|y\|) \,.
\ee
We see that $\Phi\left( B_r^{(2)}(0) \times [- \varepsilon ,  \varepsilon]\right) = U_0 \cap B_r^{(3)}(0)$ and $\lVert \Phi(y,t) \rVert = \lVert y \rVert$. Furthermore, for a fixed $y$, the map $t \mapsto \Phi(y,t)$ defined on $[-2\varepsilon,2\varepsilon]$ is a path between $U_+$ and $U_-$ inside $B_r^{(3)}(0)$ and must thus intersect $\p\L$. This path has a length of $2 \lVert y \rVert \sin^{-1}(2 \varepsilon) \le 2 \pi \varepsilon \lVert y \rVert$. Hence, as each point $v \in U_0 \cap B_r^{(3)}(0)$ is on such a path for a $y$ with $\lVert y \rVert= \lVert v \rVert$, we get $\lvert \operatorname{d}_\L(v) \rvert \le 2 \pi \varepsilon \lVert v \rVert$. Therefore, for $v \in U_0 \cap B_r^{(3)}(0)$, we get
\be
\lvert \operatorname{d}_\L(v) - v \cdot e_3 \rvert \le \lvert \operatorname{d}_\L(v) \rvert+ \lvert v \cdot e_3 \rvert \le (2\pi + 1 ) \varepsilon \lVert v \rVert\,. \label{diff quotient}
\ee
For $v \in U_{\pm} \cap B_r^{(3)}(0)$, we know that $\pm \operatorname{d}_\L(v) >0$ and only need upper and lower bounds for the distance to $\p\L$. For the lower bound, as $\p\L\cap B_{2r}^{(3)}(0) \subset U_0$, we have
\be
	\lvert \operatorname{d}_\L(v) \rvert \ge \operatorname{dist}(v,U_0) = \sqrt{1-\varepsilon^2} \lvert v \cdot e_3 \rvert - \varepsilon \lVert v^\perp \rVert  \ge \lvert v \cdot e_3 \rvert  - \varepsilon^2
	\lvert v \cdot e_3 \rvert - \varepsilon \lVert v^\perp \rVert \ge \lvert v \cdot e_3 \rvert - 2\varepsilon \lVert v \rVert\,.
\ee
For the upper bound, we just use
\be
	\lvert \operatorname{d}_\L(v) \rvert \le \lvert v \cdot e_3 \rvert+ \lvert \operatorname{d}_\L((v^\perp,0))| \le  \lvert v \cdot e_3 \rvert+2 \pi \lVert v^\perp \rVert \le  \lvert v \cdot e_3 \rvert+ 2 \pi \lVert v \rVert\,.
\ee
As the signs align, we finally get
\be
	\lvert \operatorname{d}_\L(v) - v \cdot e_3 \rvert \le 2 \pi \lVert v \rVert\,.
\ee
Thus, \eqref{diff quotient} holds for all $v \in B_{r}^{(3)}(0)$, which, by definition, says that $\operatorname{d}_\L$ is differentiable at $0$ and its differential is $e_3$. We also see that $e_3$  is orthogonal to the image of $J$ and points towards $\L^\complement$, which means that it is the outward normal vector to $\p \L$ at $0$. This finishes the proof.
\end{proof}

\begin{lemma} \label{null set1}
Let $\L\subset \R^3$ be a \textit{piecewise Lipschitz region}. Then the outward normal vector $n(v)$ exists for $\mathcal H^2$ almost every $v \in \p \L$ and the set
\be\label{A.37}
	\mathcal N \coloneqq \left \{ x^\perp \in \R^2 \colon \partial \overline{\left( \L_{x^\perp}\right) } \neq \left( \p \L\right)_{x^\perp} \right\} 
\ee
is a (two-dimensional) Lebesgue null set, where $\L_{x^\perp}$ and $\left(\p \L\right)_{x^\perp}$ are defined in \autoref{def L_x}.
\end{lemma}

\begin{proof}
We observe
\be
	 \partial \overline{\left( \L_{x^\perp}\right) } \subset  \partial \left( \L_{x^\perp}\right)  \subset \left( \p \L\right)_{x^\perp}\,.
\ee
The first inclusion is trivial. The second inclusion can be seen as follows. Let $t\in \partial \left( \L_{x^\perp}\right)$. Then for all $r>0$, $B_r^{(1)}(t)\cap \L_{x^\perp}\not=\emptyset$ and $B_r^{(1)}(t)\cap (\L_{x^\perp})^\complement\not=\emptyset$. Therefore, $B_r^{(3)}(x^\perp,t)\cap \L\not=\emptyset$ and $B_r^{(3)}(x^\perp,t)\cap \L^\complement\not=\emptyset$. Therefore, $(x^\perp,t)\in\p\L$ and $t\in \left( \p \L\right)_{x^\perp}$.

Let $\pi \colon \R^3 \to \R^2$ be the projection with $\pi(e_3)=0$ and let $\left( \Psi_{\text{pL},i}\right)_{i \in I}$ be a piecewise Lipschitz atlas of $\p \L$. For $i \in I$, we define the sets
\be
	\mathcal N_i \coloneqq \p [0,1]^2 \cup \big\{ x \in (0,1)^2 \colon D \Psi_{\text{pL},i}(x) \text{ does not exist} \big\}\,,
\ee
which are Lebesgue null sets due to Rademacher's theorem\footnote{see e.g.~\cite[Theorem 3.2]{Evans}}. Thus, the set $\bigcup_{i\in I}\Psi_{\text{pL},i}(\mathcal N_i)$ is an $\mathcal H^2$ null set, see \cite[Theorem 2.8(i)]{Evans}. Combining this with \autoref{normal vector exists lem}, we now know that the outward normal  vector $n(v)$ is well-defined for $\mathcal H^2$ every $v \in \p\L$. As $\pi \circ \Psi_{\text{pL},i} \colon [0,1]^2 \to \R^2$ is Lipschitz, this implies that $(\pi \circ \Psi_{\text{pL},i}) (\mathcal N_i)$ is a Lebesgue null set~\cite[Lemma 3.2(iii)]{Evans}. Furthermore, we define the sets
\be
	\mathcal M_i \coloneqq \big\{ x \in (0,1)^2 \colon D (\pi \circ \Psi_{\text{pL},i})(x) \text{ exists, but is not invertible} \big\}\,.
\ee
By \cite[Theorem 3.8]{Evans}, we know that $(\pi \circ \Psi_{\text{pL},i}) (\mathcal M_i)$ is a Lebesgue null set. Let
\be
	\Omega_i \coloneqq  [0,1]^2 \setminus (\mathcal N_i \cup \mathcal M_i)\,,
\ee
and 
\be 
	\mathcal M \coloneqq \p \L \setminus \bigcup_{i \in I} \Psi_{\text{pL},i}(\Omega_i)\,.
\ee
Let $v \in \p \L \setminus \mathcal M$. Hence, there is an $i \in I$ and a $y \in \Omega_i$, such that $v = \Psi_{\text{pL},i}(y)$. Thus, $D\Psi_{\text{pL},i}(y)$  exists, has full rank and does not have $e_3$ in its image. By \autoref{normal vector exists lem}, we know that $n(v)$ exists and that $n(v) \cdot e_3 \neq 0$. Thus, the function $p \colon \R \to \R$ given by $p(t)=\operatorname{d}_\L(v+te_3)$ with $\operatorname{d}_\L$ being the signed distance function to the boundary $\p\L$ has non-vanishing differential at $0$ and satisfies $p(0)=0$. Hence, $p$ changes sign at $0$, which means that $v^\parallel \in   \partial \overline{\left( \L_{v^\perp}\right) }$. Conversely, this means that for $v \in \p\L$, the property $v^\parallel \not \in  \partial \overline{\left( \L_{v^\perp}\right) }$ implies that $v \in \mathcal M$. Thus, we have for the set $\mathcal N$ defined in \eqref{A.37},
\be 
	\mathcal N \subset \pi \left(  \mathcal M\right)\,.
\ee
Finally, we observe
\begin{align}
	\mathcal N \subset  \pi\left(\bigcup_{i \in I} \Psi_{\text{pL},i}(\mathcal N_i \cup \mathcal M_i)\right) = \bigcup_{i \in I} \left((\pi \circ \Psi_{\text{pL},i}) (\mathcal N_i) \cup  (\pi \circ \Psi_{\text{pL},i}) (\mathcal M_i)\right)\,,
\end{align}
which shows that $\mathcal N$ is a Lebesgue null set.\end{proof}

\begin{lemma} \label{Lipschitz cov}
Let $\L \subset \R^3$ be a piecewise Lipschitz region, $\pi \colon \R^3 \to \R^2$ be the canonical projection and $f \colon \p \L \to \R^+$ be measurable. Let $n \colon \partial \L \to \R^3$ be the outward normal vector field, which is defined almost everywhere (see \autoref{null set1}). Then we have
\be \int_{\partial \L} \mathrm{d}{ \mathcal H^2}( v)\, f(v) \lvert n(v) \cdot e_3 \rvert  = \int_{\R^2}\mathrm{d}x \sum_{v \in \pi^{-1}(x) \cap \p \L} f(v)\,. \ee
\end{lemma}

\begin{proof}
The proof is based on a quite general form of changing variables. We use the following area-formula (see \cite[Theorem 3.9]{Evans}) with one slight modification. To this end, let $n,m \in \N$ with $n \le m$, $U \subset \R^n$ be open, $\Phi \colon U \to \R^m$ be Lipschitz continuous and let $g \colon U \to \R^+$ be measurable. Then we have the identity
\be \int_U \mathrm{d}y\, g(y) \lvert D\Phi(y) \rvert = \int_{\R^m} \mathrm{d}{\mathcal H^n}( x)\sum_{y \in \Phi^{-1}(x)} g(y) \,, \label{cov Evans}\ee 
where $\lvert D\Phi(y) \rvert^2= \det (D\Phi(y)^* D\Phi(y))$. In \cite{Evans}, this is stated for $g \in \Lp^1(U)$. However, their proof also applies to positive, measurable functions $g$ as an identity in $[0,\infty]$.
 
We cannot apply this directly to $\pi$, as it decreases the dimension and $\pi$ does not have an inverse. Thus, we have to introduce a new map. 
 
Let $(\Psi_{\text{pL},i})_{i\in I}$ be a piecewise Lipschitz atlas of $\p \L$  so that $\p\Lambda = \bigcup_{i\in I}\Psi_{\text{pL},i}([0,1]^2)$. We may assume that $\operatorname{supp}(f) \subset \Psi_{\text{pL},i} ((0,1)^2) $ for some $i \in I$. For the remainder of this proof, we write $\Psi$ for $\Psi_{\text{pL},i}$.

Now, we can apply \eqref{cov Evans} with $\Phi \coloneqq \pi \circ \Psi$ and $g \coloneqq f\circ \Psi$. Thus, we see
\begin{align}
 \int_{(0,1)^2} \mathrm{d}y\, f(\Psi(y)) \lvert D(\pi \circ \Psi)(y) \rvert &= \int_{\R^2}\mathrm{d}x \sum_{y \in (\pi \circ \Psi)^{-1}(x)} f(\Psi(y))\\
& = \int_{\R^2} \mathrm{d}x \sum_{v \in \pi^{-1}(x)\cap \p \L} f(v) \,.  \label{cov proofeq1}
 \end{align}
We used that $\Psi$ is bijective. So, we already have the right-hand side of the claim. We will apply again  \eqref{cov Evans} with the functions $\Phi \coloneqq \Psi$ and $g$ given by
\be g(y)\coloneqq f(\Psi(y))\frac{ \lvert D(\pi \circ \Psi)(y) \rvert} { \lvert D\Psi(y) \rvert}\,,\quad y\in (0,1)^2 \,. \ee
Thus, using that $\Psi$ is bijective and that the measure $\mathcal H^2$ on $\p \L$ is the $2$-dimensional Hausdorff measure, we have
\begin{align}
 \int_{(0,1)^2} \mathrm{d}y\,f(\Psi(y)) \lvert D(\pi \circ \Psi)(y) \rvert = \int_{\p \L} \mathrm{d}{ \mathcal H^2}(v)\, f(v) \frac {\lvert D(\pi \circ \Psi)(\Psi^{-1}(v)) \rvert}{\lvert D\Psi(\Psi^{-1}(v))\rvert} \,.  \label{cov proofeq2}
 \end{align}
To conclude the proof, we only need to show that the quotient of the functional determinants is given by $\lvert n(v) \cdot e_3 \rvert$ for almost every $v \in \p \L$. Let $v \in \p\L$ be such that $B \coloneqq D\Psi(\Psi^{-1}(v)) \in \R^{3 \times 2}$ is well-defined. We identify the two column vectors of the $3 \times 2$ matrix $B$ as $w_1$ and $w_2$. The image of $B$ is the tangent space to $\p \L$ at $v$. As $\Psi$ is bi-Lipschitz continuous, the matrix $B$ has full rank. The normal vector $n(v)$ is now orthogonal to the linear independent vectors $w_1,w_2$. Thus, $n(v) \lvert w_1 \times w_2 \rvert= \pm w_1 \times w_2$. 
 
As $\pi$ is linear, we have $D(\pi \circ\Psi)(\Psi^{-1}(v))= \pi B$. For the determinant of this $2 \times 2$ matrix, we get $\det(\pi B)=( w_ 1\times w_2) \cdot e_3$. For the denominator, we observe
 \be \det( B^* B)=  \lvert w_1 \rvert ^2 \lvert w_2 \rvert^2 - (w_1 \cdot w_2)^2=  \lvert w_1 \times w_2\rvert^2 \,.\ee
 In conclusion, we have
\begin{align}
\frac {\lvert D(\pi \circ \Psi)(\Psi^{-1}(v)) \rvert}{\lvert D\Psi(\Psi^{-1}(v))\rvert} = \frac {\lvert (w_1 \times w_2) \cdot e_3 \rvert} {\lvert w_1 \times w_2 \rvert} =\lvert n(v) \cdot e_3 \rvert \,.
\end{align}
In combination with \eqref{cov proofeq1} and \eqref{cov proofeq2} we have proved the statement.

\end{proof}

\begin{cl} \label{null sets2}
Let $\L\subset\R^3$ be a piecewise Lipschitz region. Then, for Lebesgue almost every $x^\perp \in \R^2$, the set $\L_{x^\perp}$ is a finite (possibly empty) union of intervals with disjoint closures.
\end{cl}
\begin{proof}
As $\mathcal H^2( \partial \L)$ is finite, see \autoref{tubular nbh small}, we have by \autoref{Lipschitz cov} (with $f=1$)
\be  \label{B.3: f=1}
\int_{\R^2} \mathrm{d} x^\perp \, \# (\partial (\L_{x^\perp})) = \int_{\partial \L} \mathrm{d} \mathcal H^2(v)\, \lvert n(v) \cdot e_3 \rvert \le \mathcal H^2(\partial \L ) < \infty\,.
\ee
This implies that the set $\p(\L_{x^\perp})\subset\R$ is finite for almost every $x^\perp$. Hence, $\L_{x^\perp}$ is almost everywhere a finite union of intervals.
 If, for some $x_0^\perp \in \R^2$, two different connected components of $\L_{x_0^\perp}$ share a boundary point, $t$, then $t \in \p( \L_{x_0^\perp}) \setminus \p  \overline{\left( \L_{x_0^\perp}\right) }$. Looking at \autoref{null set1}, we realize that this means $x_0^\perp \in \mathcal N$. Thus, we have proved the claim.
\end{proof}

\begin{lemma} \label{dist to Gamma int lem}
Let $\L \subset \R^3$ be a piecewise $\mathsf{C}^{1,\a}$ region with $\Gamma$ as in \autoref{def ps}. Then, there is a constant $C< \infty$ such that
\be \int_{\p \L}  \mathrm d \mathcal H^2(w) \lvert \ln(\operatorname{dist} (w, \Gamma) ) \rvert \le C\,. \ee
\end{lemma}

\begin{proof}
We start with
\begin{align}
 \int_{\partial \L} \mathrm{d}\mathcal H^2(w)\,\lvert \ln(\operatorname{dist}(w,\Gamma))\rvert 
&\le C \sum_{k \in \Z} (\lvert k \rvert +1) \cdot \mathcal H^2\big( \big\{ w \in \partial \L \colon2^{k-1} \le \operatorname{dist} (w,\Gamma) \le 2^k \big\} \big)\\
&\le C \sum_{k =- \infty}^{k_{\text{max}}} (\lvert k \rvert +1)\cdot \mathcal H^2\big( \big\{ w \in \partial \L \colon \operatorname{dist} (w,\Gamma) \le 2^k \big\} \big) \,. \label{piecewise C1a semifinal est}
\end{align}
 For the first step, we just bound the integrand by a step function from above. As $\L$ is bounded, the associated set is empty for $k>k_{\text{max}}$ with some finite $k_{\text{max}}$. We are left to estimate the volume of these sets. Specifically, we will show that there is an $r_0>0$ and a $C < \infty$ such that for any $r< r_0$, we have  
\be \mathcal H^2( B_r(\Gamma) \cap \partial \L ) \le Cr\,. \ee
We recall that $B_r(\Gamma)\subset \R^3$ is the $r$-neighborhood of $\Gamma$. By assumption, there is a global Lipschitz atlas $(\Psi_{\text{gL},j})_{j \in J}$ of $\p \L$, as in \autoref{def ps}.  For each $i \in I$, the set  $U_j  \coloneqq \Psi_{\text{gL},j} ((0,1)^d)\subset \p \L$ is a (relatively) open subset of the compact metric space $\p \L$ and we have $\p \L \subset \bigcup_{j \in J} U_j $. Thus, by Lebesgue's number lemma, there is a constant $r_0>0$ such that any $v\in \partial \L$ there is a $j \in J$ such that $B_{2r_0}(v) \cap \p \L \subset U_j$. 

Now, we need to understand the set $\Gamma$. We recall its definition
\be \Gamma \coloneqq \bigcup_{i \in I} \Psi_{\text{pC},i}(\p [0,1]^2)\,. \ee
Let $C_0$ be a Lipschitz constant for all $\Psi_{\text{pC},i}$'s which exists, as $I$ is finite. As $\p [0,1]^2$ is just the boundary of the unit square, there is a surjective (piecewise linear) function $\vartheta \colon [0,1] \to \p [0,1]^2$ with Lipschitz constant $4$. Let $N\in \N$ with $N>4C_0/r_0$ and $f_k \colon [0,1] \to[0,1]$ be the functions satisfying $f_k(t)=\frac{k-1+t}N$. Now, for any $1 \le k 
\le N$ and $i \in I$, we  define $g_{ik} \colon [0,1] \to \Gamma$ by $g_{ik} \coloneqq  \Psi_{\text{pC},i} \circ \vartheta \circ f_k$ and observe
\be C_{\text{Lip}}\left ( g_{ik} \right) \le 4 C_0 /N < r_0 \,. \label{gik lip est} \ee
Furthermore, $\Gamma = \bigcup_{i\in I} \bigcup_{k=1}^N g_{ik}([0,1])$. By \eqref{gik lip est}, we know  $g_{ik}([0,1]) \subset B_{r_0} (g_{ik}(0))$ and thus
\be B_r(g_{ik}((0,1)) \subset B_{2r_0}(g_{ik}(0)) \ee
for $r \le r_0$. Hence, there is an $j=j(i,k) \in J$, such that $ B_r(g_{ik}([0,1]) \cap \p \L  \subset U_{j(i,k)} $.
For any $r \le r_0$, we can estimate
 \be   \mathcal H^2( B_r(\Gamma) \cap \partial \L ) \le \sum_{i\in I} \sum_{k=1}^N \mathcal H^2\big( B_r(g_{ik}([0,1]) ) \cap  \p \L \big) \,. \ee
As $\Psi_{\text{gL},j(i,k)}$ is bi-Lipschitz, there is a constant $C$ such that
\be \mathcal H^2\big( B_r(g_{ik}([0,1]) ) \cap  \p \L \big)  \le C \left  \lvert \Psi_{\text{gL},j(i,k)}^{-1}\big( B_r(g_{ik}([0,1]) ) \cap \p\L  \big) \right \rvert \,,\ee
and
\be \Psi_{\text{gL},j(i,k)}^{-1} \left(  B_r(g_{ik}([0,1]) ) \cap  \p\L    \right)  \subset   B_{Cr}( \Psi_{\text{gL},j(i,k)}^{-1}(g_{ik}([0,1])))  \,.\ee
We now apply \autoref{tubular nbh small} with $f=\Psi_{\text{gL},j(i,k)}^{-1} \circ g_{ik}$ and $d=1$ to obtain
\be \lvert B_{Cr}( \Psi_{\text{gL},j(i,k)}^{-1}(g_{ik}([0,1])))   \rvert \le C(r+r^2) \le Cr\,, \ee
as $r < r_0$.

In conclusion, as $I$ is finite, we have
\begin{align}
\mathcal H^2(B_r(\Gamma) \cap \partial \L ) & \le  \sum_{i\in I} \sum_{k=1}^N \mathcal H^2\big( B_r(g_{ik}([0,1]) ) \cap  \p \L \big) \\
&\le  C\sum_{i\in I} \sum_{k=1}^N \left  \lvert \Psi_{\text{gL},j(i,k)}^{-1}\big( B_r(g_{ik}([0,1]) ) \cap \p\L  \big) \right \rvert \\
&\le  C\sum_{i\in I} \sum_{k=1}^N  \lvert B_{Cr}( \Psi_{\text{gL},j(i,k)}^{-1}(g_{ik}([0,1])))   \rvert  \\
&\le  C\sum_{i\in I} \sum_{k=1}^N Cr \le Cr\,. 
\end{align}
For $r_0<r<2^{k_{\text{max}}}$, we trivially arrive at the same estimate as long as $C \ge \mathcal H^2( \partial \L ) r_0^{-1}$, that is, $ \mathcal H^2(B_r(\Gamma)\cap\p\L)\le C r$ also for ``large'' $r$. 

Now, we are able to finish \eqref{piecewise C1a semifinal est} and obtain for some (finite) constant $C$
\begin{align}
 \int_{\partial \L} \mathrm{d}\mathcal H^2(w)\,\lvert \ln(\operatorname{dist}(w,\Gamma))\rvert 
\le  C  \sum_{k =- \infty}^{k_{\text{max}}} (\lvert k \rvert+1) 2^k \le C\,,
\end{align}
which was the claim.
\end{proof}

\section{Proof of \eqref{kmu integrable eq}} \label{kmu integrable section}

We observe
\begin{equation}
 \int_{\R^{m-1}}\prod_{j=1}^{m}\frac{\mathrm{d}y_j^\parallel} { \langle y_j^\parallel \rangle}  = \int_{\R^{m-1} } \mathrm dx_1^\parallel \cdots \mathrm dx_{m-1}^\parallel \prod_{j=1}^{m} \frac 1 { \langle x_j^\parallel- x_{j-1}^\parallel \rangle }\,, \label{B.4}
 \end{equation}
 where we switched back to the integration variables $x_1, \dots , x_m$ and set\footnote{The values $x_0$ and $x_m$ only matter through $x_0-x_m=0$. Thus, we can set both to $0$.} $x_0 \coloneqq x_m \coloneqq 0$. As we can see, the last expression is the $m-1$ fold convolution of $\langle \, \cdot\,  \rangle ^{-1}$ with itself evaluated at $0$. This is a job for the Fourier transform. We use the convention
 \be \mathcal F (f) (\xi) \coloneqq \lim_{R \to \infty} \int_{-R}^R \mathrm d t f(t) \,\e^{-2\pi \mathrm{i} \xi t } \,, \quad \xi \in \R\,. \ee
Thus, we have
\be \int_{\R^{m-1} } \mathrm dx_1^\parallel \cdots \mathrm dx_{m-1}^\parallel \prod_{j=1}^{m} \frac 1 { \langle x_j^\parallel- x_{j-1}^\parallel \rangle } = \int_\R \mathrm d\xi\,\mathcal F(\langle \, \cdot \,\rangle^{-1} )(\xi)^m  \, . \label{B.6}\ee
The Fourier transform of $\langle \, \cdot \, \rangle^{-1}$ can be expressed in terms of the modified Bessel function of the second kind $K_0$, see \cite[Eq.~10.32.6]{NIST:DLMF}\footnote{see \cite[Eq.~1.4.22]{NIST:DLMF} to verify their usage of an improper Riemann integral, while this paper uses Lebesgue integrals},
\be \mathcal F(\langle \, \cdot \, \rangle^{-1} )(\xi)= \lim_{R \to \infty} \int_{-R} ^R \mathrm dt\,\frac 1 {\langle t \rangle} \e^{2 \pi \mathrm{i} t \xi} = 2 \lim_{R \to \infty} \int_0^R \mathrm d t\,\frac{ \cos(2 \pi  \lvert \xi \rvert t ) } {\sqrt{t^2+1}}  = 2  K_0(2 \pi \lvert \xi \rvert ) \,. \ee
We observe that
\be 
\int_0^\infty \mathrm d \xi \,2  K_0(2 \pi \lvert \xi \rvert )  = \frac 1 2 \int_{-\infty}^\infty \mathrm d \xi\, 2 K_0(2 \pi \lvert \xi \rvert ) = \frac  1 2 \left(\langle 0 \rangle^{-1} \right)=\frac 1 2\,.\label{K0 int}
\ee
We need the (known) estimate,
\be 0 < \ln(2)- \gamma_E  < \frac 1 8 \,,\ee
where $\gamma_E$ is Euler's constant (see e.g. \cite[Eq.~5.2.3]{NIST:DLMF}). 
Using this inequality, the series representations \cite[Eq.~10.31.2, Eq.~10.25.2]{NIST:DLMF}, the harmonic series $H_n \coloneqq \sum_{k=1}^n k^{-1}\le n!$, the identity $\Gamma(n+1)=n!$ (where $\Gamma$ is the Gamma function, see e.g. \cite[Eq.~5.2.1, Eq.~5.4.1]{NIST:DLMF}) and the geometric series, we get for any $t \in (0,1)$
\begin{align}
K_0(t)&= -\left(\ln\left(\frac 1 2 t \right) + \gamma_E\right) \sum_{k=0}^\infty\frac{ \left( \frac 1 4 t^2\right)^k  }{(k!)^2} + \sum_{k=1}^\infty H_k \frac{ \left( \frac 1 4 t^2\right)^k  }{(k!)^2} \\
&<  - \left(\ln\left(\frac 1 2 t \right) + \gamma_E\right) \frac 1 {1-\frac 1 4 t^2} + \frac{t^2}{4-t^2} \,.  
\end{align}
Using the last two inequalities, we can infer
\be 2K_0(1) < 2 \left(  \frac 1 8  \frac 4 3 + \frac 1 3 \right) = 1\,.\ee
Thus, as $K_0$ is decreasing on $\R^+$ (see \cite[§10.37]{NIST:DLMF}), we have $2K_0(t)<1$ for $t>1$. For $t \in (0,1)$, we estimate using $\ln(t/2)+\gamma_E<0$ and $0<\gamma_E<1$,
\begin{align}
K_0(t) & \le - \left(\ln\left(\frac 1 2 t \right) + \gamma_E\right) \frac 1 {1-\frac 1 4 t^2} + \frac{t^2}{4-t^2} \\
& = -\ln\left( \frac 1 2 t \right)  - \gamma_E + \frac{ t^2}{4-t^2} \left( - \ln\left( t \right) + \ln(2) -\gamma_E +1 \right) \\
&< -\ln\left( \frac 1 2 t \right)  - \gamma_E  + \frac 1 3 (\sup_{t \in (0,1)}(-t^2 \ln(t) + 1+ \ln(2) -\gamma_E) ) \\
& \le -\ln\left( \frac 1 2 t \right)  - \gamma_E  + \frac 1 3 (1/(2\mathrm{e})+ 1+ \ln(2) -\gamma_E) ) < - \ln\left(\frac 1 2 t \right)\,.
\end{align}
The last step relies on a numerical computation. This can be rewritten as
\be2 K_0(2 \pi \xi) \le - 2\ln(\pi \xi) \ee
for $0<2 \pi \xi <1$. Thus, we are able to estimate 
\begin{align}
 \int_\R \mathrm d\xi\,\mathcal F (\langle \, \cdot \, \rangle^{-1} )(\xi) ^m  &= 2\int_0^\infty \mathrm d\xi\, (2K_0(2 \pi  \xi) )^m 
\\
&\le 2^{m+1}\int_0^{\frac 1{2\pi}}\mathrm d\xi\, (-\ln( \pi \xi))^m + 2\int_{\frac 1 {2\pi}} ^\infty \mathrm d\xi \, 2 K_0(2 \pi \xi) \\
 &\le   \frac 2 \pi 2^m  \int_0^1\mathrm dt\, (- \ln(t))^m + 2\int_0^\infty\mathrm d\xi\, 2K_0(2 \pi \xi)  
\\& = \frac 2 \pi 2^m m! + 1< 2^m m! \,. \label{B.12n}
 \end{align}
 The final estimate relies on $m\ge 2$ and $\pi > 3$, while the last identity is based on \eqref{K0 int} and
 \be \int_0^1 \mathrm d\xi\, (- \ln(\xi))^m  = \int_0^\infty \mathrm dt\, t^m \,\e^{-t}  = \Gamma(m+1) = m! \,.\ee
 Combining \eqref{B.4}, \eqref{B.6} and \eqref{B.12n}, we arrive at
\be  \int_{\R^{m-1}}\prod_{j=1}^{m}\frac{\mathrm{d}y_j^\parallel} { \langle y_j^\parallel \rangle} < 2^m m! \,, \ee
 which was the claim.

\section{Asymptotic expansion with order one error term}\label{Appendix C}

Our final result, \autoref{1D final asym}, in this section deals with the asymptotic expansion for a finite union of bounded intervals. That is, we assume that we have $k\in\N$ open and bounded intervals $I_1,\ldots, I_k$, whose closures are disjoint. More precisely, there exist $d_i>0$ for $1 \le i <k$ with $\sup I_i +d_i = \inf I_{i+1}$. Let $\ell_j\coloneqq |I_j|$ be the length of $I_j$ and let $\Omega \coloneqq \bigcup_{j=1}^k I_j$. The symbol $\ell$ for the length of intervals in this section has, of course, nothing to do with the index of a Landau level.

The proof of \autoref{1D final asym} is based upon two lemmata. The first lemma is per se not an asymptotic result but reduces the analysis to a single interval including an error term. The second lemma deals with the asymptotic expansion for a single interval, including an order one error term, and improves a seminal result by Landau and Widom in \cite{LW}. This is achieved by improving a certain estimate in their proof which allows for an order one error term instead of $o(\ln(L))$.  Later, Widom\cite{Widom1982} extended their result and proved that the error term is indeed of order one. This was used by Sobolev in \cite[Chapter 8]{Sob:AMS} to obtain concrete error terms. Our error term is somewhat different and fits our purposes. It is important to notice that there is still an undetermined error term of order one which is however independent of the scaling and the lengths of the intervals and depends only on the energy.

The first lemma is the following.
\begin{lemma} \label{1d multiple interval error lem}
Let $\mu >0$ and $m\in\N$. Then under the above assumptions on $\Omega$ we have
\be \left\lvert \tr \Big[ \Big( \mathds{1}_\Omega \mathds{1}[(-\mathrm{i}\nabla^\parallel)^2 \le \mu] \mathds{1}_\Omega\Big)^m - \sum_{j=1}^k \left( \mathds{1}_{I_j} \mathds{1}[(-\mathrm{i}\nabla^\parallel)^2 \le \mu] \mathds{1}_{I_j}\right)^m \Big]\right \rvert \le C \sum_{j=1}^{k-1}  \ln \Big(1+ \frac {\ell_j}{1+d_j}\Big)  \,, \ee
where $C$ is a constant depending on $m$ and $\mu$, but crucially not on $k$ or the intervals themselves.
\end{lemma}

\begin{proof} Let $Q \coloneqq \mathds{1}[(-\mathrm{i}\nabla^\parallel)^2 \le \mu]$. It is convenient to make a slight generalization by allowing $I_k$ to be any measurable set such that $d_{k-1} \coloneqq \inf(I_k)- \sup(I_{k-1} )>0$. We note that $\ell_k$  is undefined, but it is also not present in our claim. We proceed by induction with respect to the number of intervals, $k$. Once we have proved the statement for $k=2$, the statement follows for any $k$, as we can choose $I_k' \coloneqq I_k \cup I_{k+1}$. 

Hence, we just have to deal with the case $k=2$. We observe $1_\Omega= 1_{I_1}+ 1_{I_2}$. We multiply out the first term and, using $I_1 \cap I_2 = \emptyset$, we have
\begin{align}
 \left( \mathds{1}_\Omega Q \mathds{1}_\Omega\right)^m =& \sum_{j \in \{1,2\}^{\{0, \dots , m \}} } \mathds{1}_{I_{j_0}} \prod_{i=1}^m Q \mathds{1}_{I_{j_i}}\,.
 \end{align}
The two summands $j=(1,1, \dots ,1)$ and $j=(2,2, \dots,2)$ are the ones we subtract in the statement of the lemma. Hence, we have to estimate all other summands. In the case $j_0 \neq j_m$, we use $\tr AB= \tr BA$ with $A=\mathds{1}_{I_{j_0}}Q$ and $B= \prod_{i=1}^m Q \mathds{1}_{I_{j_i}}$ to conclude that the trace vanishes. We are left to estimate the terms where $j_0=j_m$ and there is an $i \in \{1,2, \dots ,m-1\}$ with $j_i \neq j_0$. In this case, we consider $i_-$ and $i_+$ as the smallest and largest such $i$ (which can be the same). Now, we  write
 \begin{align} \label{C.3}
  \mathds{1}_{I_{j_0}} \prod_{i=1}^m Q \mathds{1}_{I_{j_i}} =  \left( \mathds{1}_{I_{j_0}} Q\right) ^{i_--1} \mathds{1}_{I_{j_0}} Q \mathds{1}_{I_{j_{i_-}}} A_{i_-,i_+}     \mathds{1}_{I_{j_{i_+}}}Q \mathds{1}_{I_{j_0}} \left( Q \mathds{1}_{I_{j_0}} \right) ^{m- i_+-1} \,,
  \end{align}
where $A_{i_-,i_+}$ is the identity, if $i_-=i_+$ and a product of some operators $Q$, $\mathds{1}_{I_1}$, and $\mathds{1}_{I_2}$ otherwise. As all of the operators are projections, their operator norm can be bounded by $1$. As we are interested in the trace, we will bound the trace norm. To do so, it suffices to bound two operators in the Hilbert--Schmidt norm and all others in the operator norm. The operators we will bound in Hilbert--Schmidt norm are $\mathds{1}_{I_{j_0}} Q \mathds{1}_{I_{j_{i_-}}}$ and $ \mathds{1}_{I_{j_{i_+}}}Q \mathds{1}_{I_{j_0}}$. These operators are adjoint and hence have the same Hilbert--Schmidt norm. As $j_{i_-} \neq j_0 \neq j_{i_+}$, we know $\{j_0,j_{i_-}\}=\{j_0,j_{i_+}\}=\{1,2\}$. Thus, we are left to estimate $\lVert \mathds{1}_{I_1} Q \mathds{1}_{I_2} \rVert_2^2$. Since the operator $Q$ has integral kernel $Q(x,y)= k_\mu(x-y) = \frac{\sin(\sqrt{\mu}(x-y))}{\pi (x-y)}, x,y\in\R$, the square of the Hilbert--Schmidt norm can be easily calculated as the square of the integral of this kernel for $x \in I_1$ and $y \in I_2$. By translation invariance we may assume that $I_1=(0,\ell_1)$. By the definition of $d_1$, we know $I_2 \subset (\ell_1+d_1, \infty)$. Hence, using the estimate $|k_\mu(z)|\le C/(1+|z|)$ for some constant $C$ we get
\begin{align}
  \big|\tr \eqref{C.3}\big|\, \le  \, \lVert \mathds{1}_{I_1} Q \mathds{1}_{I_2} \rVert_2^2 &\le C \int_{I_1} \mathrm{d}x\, \int_{I_2} \mathrm{d}y\, \frac 1 {(1+ y-x)^2} \\
& \le C \int_{0}^{\ell_1} \mathrm{d}x\, \int_{\ell_1+d_1}^\infty \mathrm{d}y\, \frac1{(1+ y-x)^2}\\
&= C \int_{0}^{\ell_1} \mathrm{d}x\,\frac 1 {1+d_1+\ell_1-x} = C \ln \left( \frac{ 1+ d_1+\ell_1}{ 1+d_1} \right)\,.
 \end{align}
The number of such error terms is $2^m-1$. Thus, the error bound in $m$ is quite bad, but we only need to be good in $k$. The proof is now finished.
\end{proof}

Here is our second lemma on the mentioned improved asymptotic expansion for a single interval of Landau and Widom. This agrees with the improvement of Widom in \cite{Widom1982}. As the paper of Landau and Widom \cite{LW} is freely accessible, but the later paper by Widom \cite{Widom1982} is not\footnote{as of August 25, 2022}, we provide this different proof for the reader's convenience. We do not claim any originality. 

\begin{lemma}\label{improved Landau-Widom}
Let $\Omega \subset \mathbb R$ be an open and bounded interval of length $\ell>0$ and let $\mu>0$. Then for any $m\in\N$ and $L >0$, we have with $\mathsf{I}(m)$ explained after \eqref{def: I},
\begin{align}\label{C.7}
 \tr \big( \mathds{1}_{L\Omega} \mathds{1}[(-\mathrm{i}\nabla^\parallel)^2 \le \mu] \mathds{1}_{L\Omega}\big)^m &= \frac{\sqrt \mu}\pi L\ell +4 \,\mathsf{I}(m)\ln(1+L\ell) + O(1)\,,
\end{align}
where the order one error term is independent of $L$ and $\ell$ but depends on $\mu$.
\end{lemma}

\begin{proof}
The case $m=1$ is trivial, as the integral kernel is constant on the diagonal and only the volume term $\frac{\sqrt \mu}\pi L\ell$ appears. Thus, by linearity, it suffices to show the statement for a basis of the polynomials vanishing at $0$ and $1$. 

As $\mu$ is fixed, the result depends only on $L\ell$, which can be small or large. If $L\ell\le1$ then the trace on the left-hand side of \eqref{C.7} is bounded uniformly for these $L,\ell$ by continuity as a function of $L\ell\in[0,1]$. The same is true for the first two terms on the right-hand side of \eqref{C.7} and hence the equality holds true with an $O(1)$ error term. In the following we will assume that $L\ell>1$.

Form now on we use the same notation as in \cite{LW}, where $c$ takes the role of $L$. The last equation in the proof of their Theorem 1, where they still carry the order one error term is \cite[(18)]{LW}. Afterwards they allow for a larger $o(\ln(c))$ error term and here we take a different route. 

They consider the polynomials $(t(1-t))^n$ and $t(t(1-t))^n$ for $n \in \N$, which span all polynomials that vanish at $0$ and $1$. We proceed with these polynomials instead of $t^m$ as in the statement of our lemma. Their equation \cite[(18)]{LW} states
\be
\tr A_c [A_c(I-A_c)]^n =2 \,\tr K_c^n + O(1) = \frac 12 \tr [A_c(I- A_c)]^n + O(1)\,, \label{LW18}
\ee
where $A_c= P(0,c) Q(0,1)P(0,c)$, which is unitarily equivalent to $\mathds{1}_{[0,c]} \mathds{1}\big(-\Delta \le \frac 1 4 \big) \mathds{1}_{[0,c]}$ in our notation, and $K_c=P(1,c) Q(-\infty,0) P(-\infty, 0) Q( -\infty,0) P(1,c)$, as stated below \cite[(17)]{LW}. They state below \cite[(18)]{LW} that the integral kernel of the operator $K_c$ on $\Lp^2([1,c])$ is given for $1 \le x,y \le c$ by 
\begin{align}
f(x,y) \coloneqq \frac 1 {4 \pi^2} \int_0^\infty \frac {\mathrm{d}u} {(u+x)(u+y)} = \frac 1 {4\pi^2}\begin{cases} \frac { \ln(x)- \ln(y)} {x-y} & \mbox{ if }x \neq y\\ \frac 1 x & \mbox{ if } x=y
\end{cases}\,.
\end{align}
Let $K$ be the operator on $\Lp^2(\mathbb R^+)$ with integral kernel $f(x,y)$ for $0 <  x,y < \infty$. Thus, $K_c=P(1,c) K P(1,c)$ and $K= Q(-\infty,0) P(-\infty, 0) Q( -\infty,0)$. 
Hence, we can conclude
\be
\tr K_c^n = \int_{[1,c]^n} \mathrm{d}x\, \prod_{i=1}^n f_1(x_i, x_{i+1})\,,
\ee
with the convention $x_{n+1}=x_1$. We denote the integrand $f_n(x) \coloneqq \prod_{i=1}^n f(x_i, x_{i+1})$. It satisfies for $\lambda>0$ the homogeneity property $f_n(\lambda x)= \lambda^{-n} f_n(x)$, which indicates that we should use spherical coordinates to calculate the integral. The problem is, however, that the integration domain does not look particularly nice in spherical coordinates. Thus, we would like to change the integration domain without changing the integral too much.

The first thing to observe is that as $\ln$ is increasing, $f_n(x) \ge 0$ holds for any $n \in \N, x \in (\R^+)^n$. For any (Borel) measurable $X \subset (\R^+)^n$, we define
\be
\iota(X) \coloneqq \int_{X} \mathrm{d}x \, f_n(x)= \int_X \mathrm{d}x\,  \prod_{i=1}^n f(x_i, x_{i+1})\,.
\ee
As the integrand is non-negative, $\iota$ is a measure. We also observe that $\iota$ is invariant under the cyclic shift $(x_1, x_2, \dots, x_n) \mapsto (x_2, x_3, \dots, x_n,x_1)$. Assuming $n >1$, for $i=1, \dots, n$, we consider the set
\be
U_i \coloneqq \big\{ x \in (\mathbb R^+)^n  : x_i \le 1 \le x_{i+1}\big\}\,.
\ee
We observe $\iota(U_i)=\iota(U_1)$ by the cyclic shift property. We see
\be 
\iota(U_1)= \tr P(0,1)K P( 1, \infty) (K P(0, \infty))^{n-2} K P(0,1)\,.
\ee    
Since $P(0,1)K P(1,\infty)= P(0,1) Q(-\infty,0)P(-\infty, 0) Q(-\infty,0) P(1,\infty)$ is the operator $R$ \cite[(9)]{LW} with appropriately chosen intervals $J,M,K,N,L$ we see that this is trace class by \cite[Lemma, (L2)]{LW}. 
By the homogeneity of $f_n$, we even have $\iota( c U_1)= c^{n-n} \iota(U_1)=\iota(U_1)$. 

Next, we introduce the set 
\[ V \coloneqq \big\{x \in (\mathbb R^+)^n : \sqrt n \le \|x \| \le \sqrt{n} \, c\big\}\,,
\] 
which looks very nice in spherical coordinates. For $n=1$, we just have $V=[1,c]$. For $n>1$, we observe the chain
\be
[1,c]^n \subset V \subset [1,c]^n \cup \bigcup_{i=1}^n ( U_i \cup c U_i ) \,.
\ee
The first inclusion is trivial. We call $x_1, \dots, x_n$ the coordinates of $x$. If $x$ has both a coordinate above $\lambda$ and one below $\lambda$, then it has to be in the set $\bigcup _{i=1} ^n \lambda U_i$. Any $x \in V$ has at least one coordinate above $1$ and a coordinate below $c$. Thus, if $x \not \in [1,c]^n$, it has to have a coordinate above and below $1$ or a coordinate above and below $c$, which proves the second inclusion. These inclusions and the subbadditivity and monotonicity of $\iota$ imply that there is a constant $C_n$ (depending on $n$ but not on $c$) such that
\begin{align}
 &\iota([1,c]^n) \le \iota (V) \le \iota([1,c]^n) + C_n \qquad 
\implies \quad & \iota([1,c]^n)= \iota(V) + O(1)\,. \label{iota V+O1}
\end{align}
This holds with $O(1)$ replaced by $0$ for $n=1$. Finally, we introduce $W \coloneqq \big\{ x \in (\R^+)^n : \| x \| = \sqrt n\big\}$ with Hausdorff measure $\mathcal H^{n-1}$ and observe $V= [1,c] W = \big\{\lambda x : 1\le\lambda\le c, x\in W\big\}$. Now, we are just left to calculate
\begin{align}
\iota(V) &= \int _V \mathrm{d}x\, f_n(x) = \int_{1}^c \mathrm{d}r\,r^{n-1}  \int_{W}\mathrm{d}\mathcal H^{n-1}(x)\, f_n(r x) 
\\
& = \int_{1}^c \mathrm{d}r\, r^{n-1}  \int_{W} \mathrm{d}\mathcal H^{n-1}(x) \, r^{-n} f_n(x)  = \int_{1}^c \mathrm{d}r\,r^{-1} \int_{W} \mathrm{d}\mathcal H^{n-1}(x) \, f_n(x) 
\\
&= \ln(c) \int_{W} \mathrm{d}\mathcal H^{n-1}(x) \,f_n(x) = \tilde{C}(n) \ln(c)\,,
\end{align}
where $\tilde{C}(n)$ is the result of the surface integral. We did a change to spherical coordinates in the second step. As the integrand is positive, $\tilde{C}(n)\in (0,\infty]$ is well-defined. By \eqref{iota V+O1}, we conclude (for fixed $n$ and as $c\to\infty$)
\begin{align}
\tr K_c^n = \tilde{C}(n) \ln(c) + O(1)\,.
\end{align}
From \cite[(19)]{LW}, we know that $\tilde{C}(n)= \frac 1 {4\pi^2} \int_0^1 \mathrm{d}t\, (t(1-t))^{n-1}$. In conjunction with \eqref{LW18}, we get the improved error term $O(1)$ with the same leading term for any polynomial, which vanishes at $0$ and $1$. 

To get the claim of our lemma, we just have to replace $c$ by $2\sqrt \mu L\ell$ and then use $\ln(2\sqrt \mu L\ell) = \ln(2)+ \frac 1 2 \ln(\mu) +\ln(L\ell)=\ln(1+L\ell)+ O(1)$, which relies on $L\ell \ge 1$.
\end{proof}

Now we are in position to present and prove the main result in this section. The dependency of our error term on $\Omega$ is not just $O(1)$ as in \cite{Widom1982} but explicit in terms of the number, lengths, and distances of the constituent intervals of $\O$. Sobolev in \cite[Chapter 8]{Sob:AMS} has a similar error term, which however, does not seem to suffice for our purposes.

\begin{cl} \label{1D final asym} We assume the same conditions on the set $\Omega$ as in \autoref{1d multiple interval error lem}, $\mu>0$ and $m\in\N$. Then, with $\mathsf{I}(m)$ explained after \eqref{def: I}, we have for any $L \ge 1$,
\begin{align}
 \tr \left( \mathds{1}_{L\Omega} \mathds{1}[(-\mathrm{i}\nabla^\parallel)^2 \le \mu] \mathds{1}_{L\Omega}\right)^m &= \frac{\sqrt \mu}\pi L\lvert \Omega \rvert +4k \,\mathsf{I}(m)\ln(1+L) 
\\
&+ O\Big(k+\lvert \ln(\ell_k)\rvert+\sum_{j=1}^{k-1} \lvert \ln(\ell_j) \rvert +  \lvert \ln (d_j) \rvert \Big) \,.\label{1d multiple interval error term eq}
\end{align}
\end{cl}

\begin{proof}
For the case of a single interval, we use \autoref{improved Landau-Widom}
\begin{align}
\tr \left( \mathds{1}_{L I_j} Q \mathds{1}_{L I_j}\right)^m& =  \frac {\sqrt \mu} \pi L \ell_j +4\,\mathsf{I}(m) \ln(1+L \ell_j ) + O(1) \\
&=  \frac {\sqrt \mu} \pi L \ell_j +4\,\mathsf{I}(m) \Big[\ln(1+L)+\ln\Big(\frac{L}{1+L}\Big)+\ln\Big(\frac 1 L + \ell_j\Big)\Big] + O(1) \,. \label{C10}
\end{align}
As $L \ge 1$, we have
\begin{align}
\left \lvert \ln\left(\frac{L}{1+L}\right) + \ln\left(\frac 1 L + \ell_j\right) \right\rvert < 3+\lvert \ln(\ell_j)\rvert\,. \label{C11}
\end{align}
Next, we observe\footnote{If $ab\le 1$ then $\ln(1+ab)\le ab\le 1$ and \eqref{C.13} holds. If $ab\ge1$ then we distinguish between the case that both $a\ge1$ and $b\ge1$ and the case where one of them is smaller than $1$. In the first case \eqref{C.13} is equivalent to $\ln(1+ab)-\ln(ab) = \ln(1+1/(ab))\le1$ which holds because $\ln(1+1/(ab))\le 1/(ab)\le1$. In the remaining case we may assume $a\le1$ and $b\ge1$ (but still $ab\ge1$). Then $ab\le b/a$ and $\ln(1+ab) = \int_1^{ab} \mathrm{d}x/x \le \int_1^{b/a} \mathrm{d}x/x =\ln(b/a)$.} that for any $a>0$ and $b>0$, we have
\be\label{C.13}
\ln(1+ab) < \lvert \ln(a) \rvert + \lvert \ln(b) \rvert + 1\,.
\ee
Now, we only need to rewrite the error term from \autoref{1d multiple interval error lem} in the form we claim in this corollary. Thus, we estimate
\begin{align}
 \ln \left(1+ \frac {L\ell_j}{1+Ld_j}\right) = &\ln\left( 1+ \frac{\ell_j}{\frac 1 L + d_j } \right) \le \ln(1+\frac{\ell_j}{d_j})  <   \lvert \ln(\ell_j) \rvert + \lvert \ln(d_j) \rvert +1\,. \label{C15}
\end{align}
Once we sum the error term estimate in \eqref{C15} for $i=1, \dots, k-1$ and the one in \eqref{C11} for $i=1, \dots, k$, we arrive at the claimed error estimate in \eqref{1d multiple interval error term eq}. We also see that the sum of the main terms in \eqref{C10} for $i=1,\dots, k$ agrees with the main term in \eqref{1d multiple interval error term eq}. This finishes the proof.
\end{proof}

\section{A technical lemma on decaying functions}

This section contains a technical lemma that was useful to construct the sequence $(a_i)_{i\in\N}$ and the region $\L$ in the proof of \autoref{thm error term}.

\begin{lemma}\label{technical function existence}
Let $f \colon \R^+ \to \R^+$ be bounded and satisfy $\lim_{L \to \infty} f(L)=0$. Then there is a convex, non-increasing function $\mathrm{Env}(f) \colon \R^+ \to \R^+$ satisfying $\mathrm{Env}(f)(0)=1, \lim_{L \to \infty} \mathrm{Env}(f)(L)= \lim_{L \to \infty} f(L)/\mathrm{Env}(f)(L) =0$ and $\mathrm{Env}(f)(L) \ge C/\sqrt{\ln(2+L)}$ for some $C>0$.
\end{lemma}

\begin{proof}
The conditions on $\mathrm{Env}(f)$ only get worse if we increase $f$. Hence, we can replace $f$ by
\be \hat f(s) \coloneqq  \sup_{t >s} f(t)\,. 
\ee
This is non-increasing. To achieve the $f(L)/\mathrm{Env}(f)(L)$ condition we consider $\sqrt{\hat f}$. However, we still need to make sure that $\mathrm{Env}(f)$ is convex. For this reason, we need to consider the lower convex envelope $ \breve {\sqrt {\hat f}}$. It is given by the supremum over all convex functions below $\sqrt {\hat f}$. Another way to think of it is that the area above the graph of $\breve{ \sqrt {\hat f}}$ is the convex hull of the area above $\sqrt {\hat f}$. Finally, we define
\be \mathrm{Env}(f)(t) \coloneqq N \left( \breve{ \sqrt {\hat f}} \left( \frac t 2\right) + \frac 1 {\sqrt{ \ln(2+t)}} \right)\,,\quad t\ge0\,,\ee
where $N$ is a normalization constant to be chosen below. As the lower convex envelope and $\frac 1 {\sqrt{ \ln(2+t)}}$ are convex, so is $\mathrm{Env}(f)$.  As the lower convex envelope lies below the function, we have $\mathrm{Env}(f)(t) \le N\Big(\sqrt{\hat f (t/2)} + \frac 1 {\sqrt{ \ln(2+t)}}\Big) \to 0$ as $t \to \infty$. As $\mathrm{Env}(f)$ is convex and $\lim_{L \to \infty}\mathrm{Env}(f)(L)=0$, $\mathrm{Env}(f)$ is non-increasing.  The condition $\mathrm{Env}(f)(L) \ge C/\sqrt{\ln(2+L)}$
is trivially satisfied and implies $\mathrm{Env}(f)(0)>0$ and hence allows us to choose $N$ such that $\mathrm{Env}(f)(0)=1$. We are only left with the condition $\lim_{L \to \infty} f(L)/\mathrm{Env}(f)(L) =0$. To show this, it is sufficient to prove
\be \mathrm{Env}(f)(L) \ge C \sqrt{f(L)}\,. \ee
By the definition of the convex envelope, for any $t \ge0$, there are $0 \le t_1 \le t < t_2$ such that 
\be
\mathrm{Env}(f)(t)/N=\frac{t_2-t}{t_2-t_1} \sqrt {\hat f}\left( \frac {t_1} 2 \right) + \frac{t-t_1}{t_2-t_1} \sqrt {\hat f}\left( \frac {t_2} 2 \right) + \frac 1 {\sqrt{ \ln(2+t)}}\,.
\ee
If $t_2 \le 2t$, as $\hat f$ is non-increasing, we get $\sqrt {\hat f}\left( \frac {t_j} 2 \right)  \ge \sqrt {\hat f}(t) $ for $j=1,2$ and thus are finished. If $t_2 > 2t$, then, as $t_1 \ge0$, we have
$\frac{t_2-t}{t_2-t_1}  > \frac 12 $. Thus, it suffices to bound the first summand from below, which we already did.
\end{proof}

\end{appendix}

\end{document}